\documentclass[%
superscriptaddress,
nofootinbib,
aps,
pra,
]{revtex4-2}

\usepackage{blochsphere}
\usepackage{graphicx}
\usepackage{tabularx}
\usetikzlibrary{backgrounds,fit,decorations.pathreplacing,calc}
\graphicspath{{images/}}
\usepackage[]{xcolor} 
\usepackage{xcolor, colortbl}
\usepackage{verbatimbox}
\usepackage{array}
\usepackage{epsfig,color}
\usepackage{tikz,lipsum,lmodern}
\usepackage[most]{tcolorbox}
\usepackage{footmisc}
\usepackage{chngcntr} 

\usepackage{physics}
\usepackage{amsmath}
\usepackage{amssymb}
\usepackage{amsthm} 
\usepackage{amsfonts}
\usepackage{mathrsfs} 
\usepackage{stackrel}
\usepackage{mathtools} 

\usepackage[makeroom]{cancel}

\usepackage{listings} 
\usepackage{thmtools} 
\usepackage{thm-restate} 
\usepackage{dcolumn}
\usepackage{bm}
\usepackage[colorlinks=true, linkcolor={urlblue},citecolor={urlblue},urlcolor={urlblue}]{hyperref}
\usepackage[sans]{dsfont}
\usepackage{capt-of}
\usepackage[normalem]{ulem}
\usepackage{soul}
\usepackage[caption=false]{subfig}
\newsavebox{\imagebox}
\usepackage[capitalise]{cleveref}
\usepackage{enumitem}
\usepackage[T1]{fontenc} 
\UseRawInputEncoding
\usepackage{newtxmath}

\definecolor{nred}{rgb}{0.7,0.2,0.2}
\definecolor{npink}{RGB}{222,51,137}
\definecolor{ngreen}{RGB}{21,122,81} 
\definecolor{nblue}{RGB}{45,73,188}
\definecolor{nsteelblue}{RGB}{49,96,136}
\definecolor{npurple}{RGB}{159,40,246}
\definecolor{nblack}{rgb}{0,0,0}

\definecolor{mangotango}{rgb}{1.0, 0.51, 0.26}
\definecolor{urlblue}{RGB}{30,19,156}
\definecolor{darkgreen}{RGB}{49,113,104}
\definecolor{tssteelblue}{RGB}{70,130,180}
\definecolor{tssteelorange}{RGB}{161,57,65}
\definecolor{tsorange}{RGB}{255,138,88}
\definecolor{tsblue}{RGB}{23,74,117}
\definecolor{tsyellow}{RGB}{255,185,88}
\definecolor{tsgrey}{RGB}{200,200,200}

\lstdefinestyle{mystyle}{
}
\newcommand{\ten}{\otimes}
\newcommand{\one}{\mathds{1}} 

\DeclareFontFamily{OT1}{pzc}{}
\DeclareFontShape{OT1}{pzc}{m}{it}{ <-> s*[1.2] pzcmi7t }{}
\DeclareMathAlphabet{\mathpzc}{OT1}{pzc}{m}{it}

\def\N{\mathcal{N}}

\newcommand{\+}{^{\dagger}} 

\def\bea{\begin{eqnarray}}
\def\eea{\end{eqnarray}}

\def\bean{\begin{eqnarray*}}
\def\eean{\end{eqnarray*}}

\theoremstyle{definition}

\newtheorem*{rst*}{Result}
\newtheorem{dfn}{Definition}
\newtheorem{thm}{Theorem}
\newtheorem{lem}[thm]{Lemma}
\newtheorem{cor}[thm]{Corollary}

\newtheorem{exm}{Example}
\newtheorem*{exm*}{Example}

\newtheorem*{obs*}{Observation}

\newtheorem*{rmk*}{Remark}

\usepackage{appendix}

\begin{document}

\title{Exact and approximate conditions of tabletop reversibility: when is Petz recovery cost-free?}

\author{Minjeong Song}
\email{song.at.qit@gmail.com}

\affiliation{Centre for Quantum Technologies, National University of Singapore, Block S15, 3 Science Drive 2, 117543 Singapore}

\author{Hyukjoon Kwon}
\affiliation{School of Computational Sciences, Korea Institute for Advanced Study (KIAS), Seoul 02455, Republic of Korea}

\author{Valerio Scarani}
\affiliation{Centre for Quantum Technologies, National University of Singapore, Block S15, 3 Science Drive 2, 117543 Singapore}
\affiliation{Department of Physics, National University of Singapore, 2 Science Drive 3, Singapore 117542}

\keywords{Petz recovery map, Time-reversal, Thermal operations}

\begin{abstract}
Channels $\mathcal{N}$ that describe open quantum dynamics are inherently irreversible: it is impossible to undo their effect completely, but one can study partial recovery of the information. The Petz recovery map $\hat{\N}_{\gamma}^{(\texttt{P})}$ is a systematic construction that depends only on $\mathcal{N}$ and on a reference state $\gamma$, which will be recovered exactly. If the real input state was different from $\gamma$, the recovery is partial, with a guarantee of near-optimality. Generically, an implementation of the Petz recovery map would look very different from the implementation of the channel. It is natural to study under which conditions the two maps require similar or even identical resources. The noisy forward channel $\mathcal{N}$ is called \emph{tabletop time-reversible} for a given $\gamma$ when the corresponding Petz recovery map is realizable in such a way. First, we study the exact tabletop reversibility (TTR) conditions. We show in particular that a time-sensitive control of an ancilla system is needed for the exact TTR. Second, we present the approximate TTR conditions, which do not require such a time-sensitive control. Third, we derive Lindbladian TTR conditions under a random-time collision model, and demonstrate it with an excitation-exchange Hamiltonian. 
\end{abstract}

\maketitle

\section{Introduction}

All practical quantum information tasks confront dissipative noise from the surrounding bath over time~\cite{breuer2002oqs}. This inevitable phenomenon leads to irreversibility of noisy quantum dynamics. Nonetheless, such noisy quantum dynamics can be recovered by considering recovery channels for some quantum states. A recovery channel, known as a \emph{Petz recovery map}, has particularly proven useful in many research areas. Quantum error correction is one of its prominent applications
since a Petz recovery map for a code space is shown to be the optimal recovery when the code is perfectly correctable~\cite{barnum2002reversing}. Moreover, it performs nearly optimal compared to existing quantum error correcting codes together with recovery algorithms when the code is not too far from perfect codes~\cite{barnum2002reversing,ng2010simple,mandayam2012towards, chen2020petz,cree2022approximate,zheng2024nearqec,biswas2025noise,li2025optimality}.

It is a natural figure of merit for a recovery what states can be exactly recovered, and if approximate, how close the recovered state is from the initial state before noise. 
The former has been answered in the original paper by Petz~\cite{petz1986sufficient,petz1988sufficiency} (see also Ref.~\cite{petz2003monotonicity} for a less technical proof and Ref.~\cite{karmakar2026bayesian_von} for a recent review). It was shown that a Petz recovery map successfully recovers states when they satisfy some entropic conditions, i.e., when the relative entropy between the state to recover and a so-called reference state remains unchanged under the noisy dynamics~\cite{hayden2004qmarkov}. An approximation of this condition concerns a question as to how the changes in the relative entropy  affect the distance between the initial state and the recovered state. A Petz recovery map and its variants have then enabled an information theoretic formulation of recoverability in terms of relative entropies~\cite{wilde2015recoverability,sutter2016strengthened,junge2018universal}. They also made appearance in a plethora of studies on recoverability in the special case concerning a tripartite system in terms of the conditional mutual information~\cite{fawzi2015cmi,sutter2016universal,kim2021qmp_cmi,hu2024petzcmi,kim2024learning}. Namely, if a conditional mutual information between two subsystems conditioning the other is zero, a Petz recovery map can recover a global state from a local state of the two subsystems only and the global state is called a quantum Markov chain~\cite{hayden2004qmarkov}. Recent studies~\cite{murk2026Petz_tomography,sidajaya2026tomography} has found its use in tomography. A Petz recovery map manifests in many other quantum foundation. To name a few, it has been adopted as a crucial tool for formulating quantum fluctuation theorems~\cite{aberg18fully, kwon2019fluctuation, aw2021fluctuation,buscemi2021fluctuation, manzano2018fluctuation,scandi2025detailed} and has appeared as one of the promising candidates of quantum Bayesian retrodiction in developing quantum Bayes' rule~\cite{leifer2013towards,Tsang2022expectations,parzygnat2023axioms,parzygnat2023qbayes,aw2023quasi,scandi2023qfi,surace2023state,liu2025irreversibility, bai2025minimal,liu2025retrodictive,liu2025unifying,liu2026proper,aw2026schrodinger}. 

Apart from its fundamental importance, the implementation of Petz recovery maps is substantial to make it practically useful. There have been a few developments toward this direction. The first attempt to implement  Petz recovery maps approximately was given in Ref.~\cite{gilyen2022Petz}. Later, circuit-design approaches have been  proposed based on the resource availability in experimental setups (e.g. the circuit depth or the number of ancilla qubits) in the context of quantum error correction focusing on a code space~\cite{biswas2024circuit_Petz}. Subsequently, an experimental realization of Petz recovery maps has been recently developed for a qubit decoherence channel by designing exact quantum circuits in trapped ion platforms~\cite{png2025petz}, and in NMR platforms~\cite{singh2025petz_nmr}. 

Here, we approach this problem from a different perspective---we ask when a Petz recovery map can be implemented without requiring additional devices than the existing ones on the table. Specifically, we ask when a Petz recovery map can be realized with the same devices that were available for the forward dynamics, thus enabling `cost-free' but still nearly optimal reversal. This approach was first introduced in Ref.~\cite{aw2024role}, and the forward dynamics is said to be \emph{tabletop time-reversible} (TTR) when such a realization is possible. A recent work has studied TTR for of phase-covariant operations, revealing the interplay between irreversibility and time-translation symmetry~\cite{chen2026ttr_covariant}.

In this work, we study conditions of TTR in detail. We present the exact conditions of TTR in a directly usable form and we demonstrate it with illustrative examples. These conditions provide an insightful connection with the sufficient condition of TTR, product preservation, previously found in Ref.~\cite{aw2024role}---it provides a straightforward explanation of why the product preservation is a strictly sufficient, but not necessary condition. Yet, the exact TTR conditions require a control of an ancilla system in a timely manner in general. We thus turn to approximate TTR conditions by which we circumvent such a time-sensitive control and explicitly present the first two non-trivial approximate TTR conditions; first-order approximate TTR and second-order approximate TTR. While such approximate TTR conditions contain complex terms in general, we present simpler forms of approximate TTR conditions for some special choices of a reference state. This includes the first-order approximate implementation of the class of Petz recovery maps relevant to quantum error correction. In addition, we present Lindbladian conditions of approximate TTR, with which one can implement Petz recovery maps sequentially. These TTR-based approaches enable the implementation of Petz recovery maps, which does not require excessive resources as the unitary operations for the reversal are considered accessible and require only at most the same number of ancilla qubits as that of ancilla for the forward dynamics.

\section{Background and Definitions}

\begin{figure*}[ht]
    
    \subfloat[Dilation of Quantum Dynamics]{\includegraphics[width=0.42\textwidth]{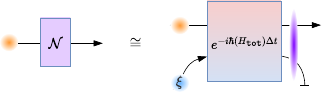}}
    \hspace{4em}
    \subfloat[Tabletop Time-reversible Dynamics]{\includegraphics[width=0.5\textwidth]{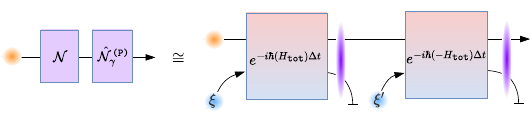}}
    
    \caption{{\bf Tabletop Time-reversibility.} The shaded in orange is the system of interest, and the shaded in blue is an environmental system which is initially uncorrelated with the system. The shaded in purple after the interaction depicts potential correlations between the system and the environment.} 
    \label{fig:ttr}
\end{figure*}

We consider linear quantum channels. Any linear dynamics $\mathcal{N}$ of a quantum system \texttt{S} can be obtained by i) appending an ancillary system \texttt{E} in a suitable state $\xi$, ii) acting with a joint unitary transformation $U$ on the total systems \texttt{SE}, and iii) discarding the ancillary system \cite{wilde2013qit}:
\begin{eqnarray}
    \N(\bullet) = \Tr_{\texttt{E}}\left(U_{\texttt{SE}}(\bullet\ten\xi_{\texttt{E}}) U_{\texttt{SE}}^\dagger\right). \label{eq:exp_friend}
\end{eqnarray}
In general, a channel can have multiple such \textit{dilations} $(U,\xi)$. There is a mathematically canonical choice, the Stinespring dilation, where $\xi_{\texttt{E}}$ is pure in a space of dimension $d_E=d_S^2$: this dilation is unique up to additional local unitaries on $E$. However, it is not necessarily the most natural. Even at the theory level, in thermal operations (see below) the natural dilation uses the thermal state of the environment as ancilla. In an experiment, one would model the state of the environment without having to reduce it to a pure state of a fixed dimension. We shall come back to this. 

The Petz recovery map of a channel $\mathcal{N}$ is defined as follows. 

\begin{dfn}[Petz recovery map]
    Given a channel $\N$, its \textit{Petz recovery map} for a reference state $\gamma$ is defined as
\begin{eqnarray}\label{eq:defpetz}
    \hat{\N}_{\gamma}^{(\texttt{P})}(\bullet)\equiv\sqrt{\gamma}\N^\dagger\qty(\frac{1}{\sqrt{\N(\gamma)}}
    \bullet\frac{1}{\sqrt{\N(\gamma)}}
    )\sqrt{\gamma}.
\end{eqnarray}    
Following a nomenclature from Bayesian retrodiction, we will call $\gamma$ a prior state (another frequently used name is reference state).
\end{dfn}

The need for recovery arises naturally in the noisy regime, particularly when the initial state is unknown to the experimentalist; otherwise, a party who receives the noisy state and already knows the initial state could simply reprepare it, making any recovery unnecessary. In general, no recovery map can fully undo the noise. Nevertheless, the Petz recovery map provides a remarkably good approximate correction, and its near-optimality has been established with respect to various figures of merit~\cite{barnum2002reversing,ng2010simple,mandayam2012towards,chen2020petz,cree2022approximate,zheng2024nearqec,biswas2025noise,li2025optimality,chen2025optical_Petz}. In quantum error correction, in particular, this near-optimality is measured by the entanglement fidelity~\cite{schumacher1996ent_fidelity,nielsen1996ent_fidelity}.

Specifically, suppose the initial state is prepared in some orthonormal basis $\ket{i}$ with probability $q_i$, so that the preparation ensemble $\mathcal{E}=\{q_i,\ket{i}\}$ has the average initial state  $\rho_\texttt{int}=\sum_i q_i\ket{i}\bra{i}$. This state then undergoes noisy dynamics $\mathcal{N}$. After applying a recovery map $\mathcal{R}$, its recovery performance can be measured by the \emph{entanglement fidelity}
\begin{align}
    F_\texttt{ent}(\rho_\texttt{int}, \mathcal{R}\circ \mathcal{N} ) \equiv \bra*{\Psi_{\rho_\texttt{int}}} \left(id \ten \mathcal{R}\circ \mathcal{N}\right)\left(\op*{\Psi_{\rho_\texttt{int}}}\right) \ket*{\Psi_{\rho_\texttt{int}}}, \label{eq:ent_fidelity}
\end{align}
where $\ket{\Psi_\rho}$ is the purification of $\rho$, defined as $\sum_i \sqrt{\rho}\ket{i}\otimes \ket{i}$ with $\{\ket{i}\}_i$ the eigenbasis of $\rho$. In Ref.~\cite{barnum2002reversing}, the Petz recovery map is shown to be near-optimal in the sense that
\begin{align}
    F_\texttt{ent}(\rho_\texttt{int}, \mathcal{N}^\texttt{(P)}_{\gamma}\circ \mathcal{N} ) \ge \left(F_\texttt{ent}(\rho_\texttt{int}, \mathcal{R}_\texttt{opt}\circ \mathcal{N} )\right)^2 \quad \text{when} \,\,\gamma = \rho_\texttt{int}, \label{eq:near_optimal}
\end{align}
where $\mathcal{R}_\texttt{opt}$ denotes the optimal recovery map that maximizes the entanglement fidelity in Eq.~\eqref{eq:ent_fidelity} for fixed initial state $\rho_\texttt{int}$ and noisy channel $\mathcal{N}$. In this work, we focus on its physical implementation and on the question of when it can be `easily' implemented, rather than further demonstrating the near-optimality of the Petz recovery map.

Notably, choosing a pure prior state always makes the implementation of a Petz recovery map trivial---it becomes an erasure channel that outputs the prior state, regardless of its input. That is, if one holds a too strong belief in the prior state, the reversal always outputs the prior state. In what follows, we will consider full rank priors only so as to avoid such trivial cases. In addition, we assume that the prior state after the forward dynamics remains full rank, as otherwise the Petz recovery maps are ill-defined or not uniquely defined \cite{aw2024role}. 

Though written as a sequential composition \eqref{eq:defpetz} of three non-trace-preserving maps, the Petz recovery map should not be implemented as such a composition, but as a single trace-preserving channel. For that, it is very useful (possibly essential) to find a dilation as in Eq.~\eqref{eq:exp_friend}:
\begin{eqnarray}
\hat{\N}_{\gamma}^{(\texttt{P})}(\bullet) = \Tr_{\texttt{E}}\left(V_{\texttt{SE}}(\bullet\ten\xi'_{\texttt{E}}) V_{\texttt{SE}}^\dagger\right) \label{eq:petz_friend}
\end{eqnarray} where the dilation $(V,\xi')$ depends on $\gamma$. Dilations of Petz recovery maps have been studied only in a few cases, a prominent one being that of thermal operations from quantum thermodynamics. The channel 
\begin{align*}
    \mathcal{T}(\bullet)=\Tr_{\texttt{E}}\left[ U_{\texttt{SE}} (\bullet\ten \rho_{\texttt{th}}^{(\beta)}(H_{\texttt{E}})) U_{\texttt{SE}}^\dagger \right],
\end{align*}
is a thermal operation if $\comm{U_{\texttt{SE}}}{H_{\texttt{S}}+H_{\texttt{E}}}=0$, with $H_{\texttt{S}}$ and $H_{\texttt{E}}$ the free Hamiltonians of the system and the environment, and $\rho_{\texttt{th}}^{(\beta)}(H_{\texttt{E}})=\frac{e^{-\beta H_{\texttt{E}}}}{\Tr e^{-\beta H_{\texttt{E}}}}$ is the Gibbs state of the environment with $\beta=\frac{1}{k_{\texttt{B}}T}$. In this case, the Petz recovery map with prior $\gamma=\rho_{\texttt{th}}^{(\beta)}(H_{\texttt{S}})$ consists of inverting time in $U_{\texttt{SE}}$ while interacting with the same environment \cite{alhambra2018work}: 
\begin{align*}
    \hat{\mathcal{T}}_{\rho_{\texttt{th}}^{(\beta)}(H_{\texttt{S}})}^{(\texttt{P})}(\bullet)= \Tr_{\texttt{E}}\left[ U_{\texttt{SE}}^\dagger (\bullet\ten \rho_{\texttt{th}}^{(\beta)}(H_{\texttt{E}})) U_{\texttt{SE}} \right].
\end{align*}
Here, implementing the Petz recovery map
does not require building additional devices that carry out a joint unitary for the reversal--- it simply uses the existing devices from the forward dynamics. 

Generalizing from this example, the following definition was introduced in Ref.~\cite{aw2024role}:
\begin{dfn} [Tabletop reverse map]
    Given the dilation $(U,\xi)$ of the channel $\mathcal{N}$, its \textit{tabletop reverse map} with an ancilla in a state $\xi'$ is defined as
\begin{eqnarray}
    \hat{\N}_{\xi'}^{(\texttt{T})}(\bullet)\equiv\Tr_{\texttt{E}}\left(U_{\texttt{SE}}^\dagger(\bullet\ten\xi'_{\texttt{E}}) U_{\texttt{SE}}\right).
\end{eqnarray} This is depicted in \cref{fig:ttr}. Then the dilation $(U,\xi)$ of the channel $\N$ is \emph{tabletop time-reversible} (TTR) for a prior $\gamma$, if its Petz recovery map is a tabletop reverse map with $\xi'=\xi'_\gamma$,
\begin{eqnarray}
\hat{\N}_{\gamma}^{(\texttt{P})}=\hat{\N}^{(\texttt{T})}_{\xi'_\gamma}, \label{eq:TTRdef}
\end{eqnarray} i.e.,~if ~$V_{\texttt{SE}}=U_{\texttt{SE}}^\dagger$ in Eq.~\eqref{eq:petz_friend}.
\end{dfn}

Fixing the dilation matches the spirit of the definition, which is that the reverse process should be implementable with the same devices as the forward one. Besides, if the dilation is left arbitrary, the TTR condition becomes trivial. Indeed, suppose that $(U,\xi)$ is a dilation of $\mathcal{N}$ and $(V,\xi')$ is a dilation of $\hat{\N}_{\gamma}^{(\texttt{P})}$. Defining the unitary $W:=U\ten \op{0} + {V}^\dagger\ten \op{1}$, we see that $(W, \xi_1)$ is another dilation of $\mathcal{N}$ for $\xi_1:=\xi\ten \op{0}$, and $(W^\dagger, \xi_2)$ is a dilation of $\hat{\N}_{\gamma}^{(\texttt{P})}$ for $\xi_2:=\xi'\ten \op{1}$. Thus, given dilation $(W, \xi_1)$, $\mathcal{N}$ becomes TTR for $\gamma$ with regard to $\xi_2$.

Having fixed the dilation $(U,\xi)$ and chosen a prior $\gamma$, the existence of $\xi'$ with which $(U,\xi)$ becomes TTR can be checked by a semidefinite program (SDP)~\cite{skrzypczyk2023sdp}, the feasibility SDP~\cite{Note_sdp}
\begin{eqnarray}
\begin{aligned}
    \text{Minimize}& \quad 0 \\ 
    \text{subject to}& \quad \xi' \ge 0\\
    & \,\, \Tr\xi'=1\\
    & \,\, \chi_{\xi'}^{(\texttt{T})} = \chi_{\gamma}^{(\texttt{P})}\,,
\end{aligned}    
\end{eqnarray}
where $\chi_{\gamma}^{(\texttt{P})}$ and $\chi_{\xi'}^{(\texttt{T})}$ are the Choi matrices~\cite{choi75} of $\hat{\N}_{\gamma}^{(\texttt{P})}$ and $\hat{\N}^{(\texttt{T})}_{\xi'}$, respectively.

In this work, we aim at analytical characterizations. What is known comes from Ref.~\cite{aw2024role}. This work first identified \textit{product preservation} as a sufficient condition of tabletop reversibility. Product preservation holds for thermal operations and Gibbs states at the same temperature:  
\begin{align}
    U_{\texttt{SE}}\left(\rho_{\texttt{th}}^{(\beta)}(H_{\texttt{S}})\ten \rho_{\texttt{th}}^{(\beta)}(H_{\texttt{E}})\right) U^\dagger_{\texttt{SE}} = \rho_{\texttt{th}}^{(\beta)}(H_{\texttt{S}})\ten \rho_{\texttt{th}}^{(\beta)}(H_{\texttt{E}})\,.
\end{align}
By generalizing this, a unitary $U_{\texttt{SE}}$ is said to be product preserving with respect to $\gamma\ten\xi$ if it satisfies 
\begin{align}
    U_{\texttt{SE}}(\gamma\ten \xi) U^\dagger_{\texttt{SE}} = \gamma'\ten \xi' \label{eq:product_preserving}
\end{align}
for some states $\gamma',\xi'$. The conditions for product preservation were then fully characterized for the case when both the system and the environment are one qubit. In the same paper, some counterexamples showed that product preservation is not a necessary condition: there are tabletop time-reversible dynamics that do not arise from product preservation.

Throughout this paper, we will use the following spectral decompositions as standard notations, unless otherwise specified:
\begin{align}
	\begin{aligned}
		\gamma\stackrel{s.d}{=}&\sum_n r_n \op*{\lambda_n}, \quad \gamma'=\N(\gamma)\stackrel{s.d}{=}\sum_m r'_m \op*{\lambda'_m},\\
		\xi\stackrel{s.d}{=}&\sum_k p_k \op*{e_k}, \quad \xi'\stackrel{s.d}{=}\sum_j p'_j \op*{e'_j}.   
	\end{aligned}\label{eq:params}
\end{align}
With these notations, we introduce the choice of Kraus operators that we shall use later. Using the shorthand $\ket{\cdot}_{\texttt{E}}$ to represent $\one_{\texttt{S}}\ten\ket{\cdot}_{\texttt{E}}$. for every orthonormal basis $\{\ket*{f_j}\}_j$ of system $E$ we can write Eq.~\eqref{eq:exp_friend} as 
\begin{align}
\begin{aligned}
    \N(\bullet) &= \sum_{j,k} \left(\sqrt{p_k} \bra*{f_j}_{\texttt{E}}U_{\texttt{SE}}\ket*{e_k}_{\texttt{E}}\right) \bullet \left(\bra*{e_k}_{\texttt{E}}U_{\texttt{SE}}^\dagger \ket*{f_j}_{\texttt{E}}\sqrt{p_k}\right),
\end{aligned}
\end{align}
with Kraus operators $N_{jk} = \sqrt{p_k}\bra*{f_j}_{\texttt{E}}U_{\texttt{SE}}\ket*{e_k}_{\texttt{E}}$. This representation is not unique because the $\{\ket*{f_j}\}_j$ can be chosen arbitrarily. For us, it is convenient to choose $\ket*{f_j}:=\ket*{e'_j}$ the eigenvectors of $\xi'$, so that
	\begin{align}
		N_{jk} = \sqrt{p_k}\bra*{e'_j}_{\texttt{E}}U_{\texttt{SE}}\ket*{e_k}_{\texttt{E}}.
	\end{align}
Similarly, the Kraus operators of the Petz recovery map $\hat{\N}_{\gamma}^{(\texttt{P})}$ and of the tabletop reverse map $\hat{\N}^{(\texttt{T})}_{\xi'}$ can be chosen as
\begin{align}
	\begin{aligned}
		\hat{N}_{jk}^{(\texttt{P})} =&
		\sqrt{p_k}\sqrt{\gamma}\bra*{e_k}_{\texttt{E}}U^\dagger_{\texttt{SE}}\ket*{e'_j}_{\texttt{E}}\frac{1}{\sqrt{\gamma'}}, \qand 
		\hat{N}_{jk}^{(\texttt{T})} =&
		\sqrt{p'_j}\bra*{e_k}_{\texttt{E}}U^\dagger_{\texttt{SE}}\ket*{e'_j}_{\texttt{E}} \label{eq:kraus}
	\end{aligned}
\end{align}
respectively. 

With this notations, we can look back at the product preservation condition \eqref{eq:product_preserving}. This condition, recast as $(\gamma\ten \xi)U_{\texttt{SE}}^\dagger  = U_{\texttt{SE}}^\dagger (\gamma'\ten \xi')$, implies 
\begin{align}
		\sqrt{p_k}\sqrt{\gamma}\bra*{e_k}_{\texttt{E}}U^\dagger_{\texttt{SE}}\ket*{e'_j}_{\texttt{E}}  =  \sqrt{p'_j}\bra*{e_k}_{\texttt{E}}U^\dagger_{\texttt{SE}}\ket*{e'_j}_{\texttt{E}} \sqrt{\gamma'}\;
		&\Leftrightarrow\; \hat{N}_{jk}^{(\texttt{P})} = \hat{N}_{jk}^{(\texttt{T})}, \, \forall \, j,k. 
\end{align} Since the Kraus operators of the two maps are equal element-wise, we recover the fact that product preservation implies TTR. Recall that the converse is known to be false~\cite{aw2024role}. We can now move to an explicit necessary and sufficient condition for TTR.

\section{Exact TTR}

\subsection{Explicit condition for TTR}

We are going to prove the following:
\begin{thm}\label{rst:exactTTR}
	A channel $\mathcal{N}$ is TTR with regard to $\xi'$ for $\gamma$ if and only if $\forall m_1,m_2,n_1,n_2$,
	\begin{align}
			\sum_{jk} \left(p_k\sqrt{r_{n_1}r_{n_2}}-p'_j\sqrt{r'_{m_1}r'_{m_2}}\right)\Phi_{m_1 j\leftarrow n_1k}\Phi^*_{m_2 j\leftarrow n_2k} = 0, \label{eq:iff}
	\end{align}
	where $c^*$ represents the complex conjugate of $c$ and with the transition matrix
	\begin{align}
		\Phi_{mj\leftarrow nk}\equiv\bra*{\lambda'_m}_{\texttt{S}}\bra*{e'_j}_{\texttt{E}}U_{\texttt{SE}}\ket*{\lambda_n}_{\texttt{S}}\ket*{e_k}_{\texttt{E}}.     
	\end{align} The parameters and vectors are defined in Eq.~\eqref{eq:params}.
\end{thm}

\begin{proof}
	 We can rewrite the TTR condition $\hat{\mathcal{N}}^{(\texttt{P})}_{\gamma}=\hat{\mathcal{N}}^{(\texttt{T})}_{\xi'}$ with the Kraus operators of the channels as 
	\begin{align*}
		\begin{aligned}
			\sum_{jk}\hat{N}_{jk}^{(\texttt{P})}(\ketbra*{\lambda'_{m_1}}{\lambda'_{m_2}})\hat{N}_{jk}^{(\texttt{P})\dagger}=&\sum_{jk}\hat{N}_{jk}^{(\texttt{T})}(\ketbra*{\lambda'_{m_1}}{\lambda'_{m_2}})\hat{N}_{jk}^{(\texttt{T})\dagger}
		\end{aligned}
	\end{align*}
	for any $m_1,m_2$ because $\{\ket*{\lambda'_{m}}\}_m$ spans the Hilbert space of the image of the channel $\N$. Together with the 
	Kraus representations as in Eq.~\eqref{eq:kraus},
	multiplying this with $\bra*{\lambda_{n_1}}(\cdot)\ket*{\lambda_{n_2}}$, 
	it becomes obvious that Eq.~\eqref{eq:iff} holds iff $\N$ is TTR: for any $m_1,m_2,n_1,n_2$,
	\begin{align*}
    \begin{aligned}
        & \sum_{jk} \left(p_k\frac{\sqrt{r_{n_1}r_{n_2}}}{\sqrt{r'_{m_1}r'_{m_2}}}-p'_j\right)\Phi^*_{m_1 j\leftarrow n_1k}\Phi_{m_2 j\leftarrow n_2k}=0,
    \end{aligned}
	\end{align*} 
    which is equivalent to Eq.~\eqref{eq:iff}. This equation condenses in a single matrix condition the fact that the actions of the two channels are equal, as it reads $\texttt{Super}(\{\hat{N}_{jk}^{(\texttt{P})}\}_{jk})=\texttt{Super}(\{\hat{N}^{(\texttt{T})}_{jk}\}_{jk})$, where $\texttt{Super}(\{E_k\}_k)\equiv \sum_{k} E_k \ten E_k^*$. 
\end{proof}

One of the difficulties of the TTR problem arises from the fact that both the eigenvalues and the eigenbasis of $\xi'$ must be determined so that, for a chosen prior, the dynamics is TTR. The exact condition of Eq.~\eqref{eq:iff} is convenient in this regard,
because it is divided into two parts: one fixed by the choice of eigenbases and the
other by the eigenvalues. The eigenbases of $\gamma,\gamma',\xi$ and $\xi'$ appear only as a encoded form in
the transition matrix $\Phi_{mj\leftarrow nk}$, while their eigenvalues appear as weights. 

Consequently, once the eigenbasis of $\xi'$ is fixed, only its weights $p'_j$'s remain
unknown. Representing them in a vector $(\vec x)_j\equiv p_j'$ turns the
verification of exact TTR into an elementary matrix equation
\begin{align}
    M \Vec{x} = \Vec{y}, \label{eq:matrix}
\end{align}
with a \emph{fixed} matrix $M$ whose row is indexed by $m_1m_2n_1n_2$ and column indexed by $j$ and a \emph{fixed} vector $\Vec{y}$ whose row is also indexed by $m_1m_2n_1n_2$, defined as
\begin{align}
    (M)_{m_1 m_2 n_1 n_2, j}\equiv \sum_k \Phi_{m_1 j\leftarrow n_1 k}\Phi^*_{m_2 j\leftarrow n_2 k} 
    \qand 
    (\Vec{y})_{m_1m_2n_1n_2}\equiv \frac{\sqrt{r_{n_1} r_{n_2}}}{\sqrt{r'_{m_1} r'_{m_2}}}
    \sum_{jk} p_k\,\Phi_{m_1 j\leftarrow n_1 k}\Phi^*_{m_2 j\leftarrow n_2 k}.
\end{align}
We will often refer to this matrix $M$ as the \emph{TTR matrix} and the Eq.~\eqref{eq:matrix} as the \emph{TTR matrix equation}. The basic linear algebra tells us that the TTR matrix equation has a solution if and only if the TTR matrix $M$ and the augmented matrix $[M \vert \Vec{y}]$ have the same rank. Stated differently, there exists a solution if $\Vec{y}$ lies in the image of $M$, which is the span of the column vectors of $M$. The dimension of the ancilla system is the number of column vectors of $M$, and additional columns can enlarge the image of $M$. A higher dimensional ancilla system can increase the chance that the exact TTR condition is met, since it provides more free parameters, in this case the $p'_j$'s, to balance while respecting the completeness condition $\sum_j p'_j=1$. In \cref{sec:application}, we will see explicitly how a higher dimensional ancilla system can increase the chance that the exact TTR condition is met. 

A Petz recovery map has a connection with quantum fluctuation theorems~\cite{aberg18fully, kwon2019fluctuation, aw2021fluctuation,buscemi2021fluctuation, manzano2018fluctuation,scandi2025detailed}, notably Tasaki's two-measurement setup \cite{tasaki00jarzynski} and its generalizations beyond the case of a unitary channel between the two measurements \cite{buscemi2021fluctuation}. This connection appears when considering the eigenstates, $\Pi_{n}\equiv \op*{\lambda_{n}}$ and $\Pi'_{m}\equiv \op*{\lambda'_{m}}$, in the following equality
\begin{align}
\begin{aligned}
    \Tr\left( \hat{\N}^{(\texttt{P})}_{\gamma} ( \Pi'_{m} ) \Pi_{n}  \right) &= \Tr\left(\sqrt{\gamma} \N^\dagger (\frac{1}{\sqrt{\gamma'}}  \Pi'_{m} \frac{1}{\sqrt{\gamma'}} ) \sqrt{\gamma} \Pi_n \right) \\
    &= \frac{r_n}{r'_m} \Tr\left( \N (\Pi_{n} ) \Pi'_{m}  \right). \label{eq:bayes}
\end{aligned}
\end{align}
Now,
\begin{align}
    \begin{aligned}
\Tr\left( \N (\Pi_{n} ) \Pi'_{m}  \right)\equiv \Pr(\ket*{\lambda_n} \stackrel{\N}{\rightarrow} \ket*{\lambda'_m} ) \qand \Tr\left(\hat{\N}^{(\texttt{P})}_{\gamma} ( \Pi'_{m} ) \Pi_{n}  \right)\equiv\Pr\Big(\ket*{\lambda'_m} \stackrel{\hat{\mathcal{N}}^{(\texttt{P})}_{\gamma}}{\rightarrow} \ket*{\lambda_n} \Big)
    \end{aligned}
\end{align}
represent respectively the probability $p(m|n)$ of going from $\ket*{\lambda_n}$ to $\ket*{\lambda'_m}$ via $\mathcal{N}$, and the probability $q(n|m)$ of the reverse transition via $\hat{\mathcal{N}}^{(\texttt{P})}_{\gamma}$. Thus Eq.~\eqref{eq:bayes} reads formally $r_n p(m|n)=r_m'q(n|m)$, showing that $\hat{\mathcal{N}}^{(\texttt{P})}_{\gamma}$ implements classical Bayesian retrodiction on the two-measurement scheme. Such a detailed-balance~\cite{thomsen1953principle} between probabilities underpins a fluctuation theorem, and generalizing this idea to quantum theory has been widely explored~\cite{ campisi2011fluctuation, esposito2009fluctuation}.
Including the environment systems, on the other hand, we instead have
\begin{align}
    \sum_{jk} r_n p_k \Pr(\ket*{\lambda_n}\ten \ket*{e_k} \stackrel{U}{\rightarrow} \ket*{\lambda'_m} \ten \ket*{e'_j} ) = \sum_{jk} r'_m p'_j \Pr(\ket*{\lambda'_m}\ten \ket*{e'_j} \stackrel{U^\dagger}{\rightarrow} \ket*{\lambda_n} \ten \ket*{e_k} ), \, \forall \, n,m.
\end{align}
This suggests that the TTR condition can be viewed as a detailed-balance condition between $\gamma$ and $\gamma'$ in the dilation picture. From this perspective, the TTR condition could provide a useful framework for exploring the connection between fluctuation theorems based on quantum retrodiction~\cite{aw2021fluctuation,buscemi2021fluctuation,leifer2013towards,parzygnat2023qbayes} and system-bath correlations, for example, in the context of exchange fluctuation theorems~\cite{jarzynski2004exchange_fluctuation, akagawa2009exchange_fluctuation}. Meanwhile, Kubo-Martin-Schwinger (KMS) detailed-balance condition requires $\N=\hat{\N}^{(\texttt{P})}_\gamma$ and that $\gamma$ is a steady state of $\N$, turning Eq.~\eqref{eq:bayes} into a classical detailed-balance between transition probabilities induced by the eigenvalues of $\gamma$. Imposing both conditions implies  $\N=\hat{\N}^{(\texttt{T})}_{\xi'}$ and $\gamma'=\N(\gamma)=\gamma$. These conditions offer a stepping stone for studying a detailed-balance between the eigenvalues of $\gamma$ in the dilation picture where we invoke the true reversal $U^\dagger$ of $U$ instead of a Bayesian retrodiction. We leave a more in-depth investigation of these connections for future work.

\subsection{Exploiting the condition: three case studies}

Next, we demonstrate the use of \cref{rst:exactTTR} with case studies. In all of them, the prior state is chosen as a steady state of the channel $\mathcal{N}(\gamma)=\gamma$. More examples of TTR channels can be found in Ref.~\cite{li2026experimentalTTR}, alongside with their experimental demonstration.

\begin{exm}[Factorizable channels]
    Consider a $d$-dimensional unital channel $\mathcal{N}$ defined by $(U,\xi)$ with an ancilla $\xi=\frac{\one}{d}$ for some unitary $U$. $\mathcal{N}$ is tabletop time-reversible with regard to $\xi'=\xi$ when $\gamma$ is a steady state. Here $\one$ denotes the identity operator. \label{exm:factorizable}
\end{exm}

If $\gamma=\frac{\one}{d}$, TTR holds because of the trivial product preservation $U(\one\otimes\one)U^\dagger=\one\otimes\one$. The same conclusion also follows from the TTR matrix equation. Since $\gamma$ is a steady state and $\xi=\gamma=\frac{\one}{d}$, all relevant eigenvalues are $\frac{1}{d}$, i.e., $r_n=r'_m=p_k=\frac{1}{d}$. The output vector $(\vec{y})_{m_1m_2n_1n_2}$ in the TTR matrix equation then evaluates to
\begin{align}
    (\vec{y})_{m_1m_2n_1n_2}&\equiv \frac{\sqrt{r_{n_1} r_{n_2}}}{\sqrt{r'_{m_1} r'_{m_2}}}
    \sum_{jk} p_k\,\Phi_{m_1 j\leftarrow n_1 k}\Phi^*_{m_2 j\leftarrow n_2 k}\\
    &=\sum_{jk} \frac{1}{d}\,\Phi_{m_1 j\leftarrow n_1 k}\Phi^*_{m_2 j\leftarrow n_2 k}\\
    &=\sum_{j} \frac{1}{d} (M)_{m_1m_2n_1n_2,j},
\end{align}
so that the TTR matrix equation $M\Vec{x}=\Vec{y}$ is satisfied by  $(\vec{x})_j=\frac{1}{d}$ corresponding to $\xi'=\frac{\one}{d}$ (thus its eigenbasis is irrelevant).

Some unital channel may have other steady states $\gamma$. To verify that TTR also holds for them, we first observe that $\hat{\mathcal{N}}^{(\texttt{P})}_{\gamma}=\mathcal{N}^\dagger$:
\begin{align}
    \begin{aligned}
        \hat{\mathcal{N}}^{(\texttt{P})}_{\gamma}(\bullet) &\equiv\sqrt{\gamma}\sum_{jk} N_{jk}^\dagger \qty(\frac{1}{\sqrt{\N(\gamma)}}
    \bullet\frac{1}{\sqrt{\N(\gamma)}}
    )N_{jk}\sqrt{\gamma}\\
    &=\sum_{jk}\sqrt{\gamma} N_{jk}^\dagger \qty(\frac{1}{\sqrt{\gamma}}
    \bullet\frac{1}{\sqrt{\gamma}}
    )N_{jk}\sqrt{\gamma}\\
    &=\sum_{jk}\sqrt{\gamma} \frac{1}{\sqrt{\gamma}} N_{jk}^\dagger \qty(
    \bullet
    )N_{jk}\frac{1}{\sqrt{\gamma}}\sqrt{\gamma} = \mathcal{N}^\dagger(\bullet).
    \end{aligned}
\end{align}
The second equality follows because $\gamma$ is a steady state; the third equality follows from Theorem 4.25 in Ref.~\cite{watrous2018qit}, stating that when a channel is unital, a state $\gamma$ commutes with every Kraus operator $N_{jk}$ of the channel (and thus with their adjoint $N_{jk}^\dagger$) if and only if $\gamma$ is a steady state. Then we complete the verification by noting that $N_{jk}^\dagger=N_{jk}^{(\texttt{T})}$, since $N_{jk}=\frac{1}{d}\bra*{j}_{\texttt{E}} U\ket*{k}_{\texttt{E}}$ and $N_{jk}^{(\texttt{T})}=\frac{1}{d}\bra*{k}_{\texttt{E}} U^\dagger \ket*{j}_{\texttt{E}}$ for some orthonormal basis $\{\ket{j}\}$. As a side note, from $\hat{\mathcal{N}}^{(\texttt{P})}_{\gamma}=\mathcal{N}^\dagger$ it follows almost immediately that, for a unital channel, the Petz recovery map for a steady state achieves optimal recovery according to entanglement fidelity.\\

In summary, time-reversal of unital channel with a steady state as prior can be implemented by a simple backward process, such as the axiomatic time-reversal introduced in Ref.~\cite{chiribella2022timeflip} (called input-output inversion therein). This example was simple enough, that we did not need to use \cref{rst:exactTTR}. We use it in the following cases.

\begin{exm}[Controlled-$X$ gates]
    Let $X,Z$ be Pauli operators. Consider a dynamics $\mathcal{N}$ defined by $U=\op{0}\ten\one+\op{1}\ten X$ and an ancilla $\xi$ such that $\comm{\xi}{Z}=0$.  $\mathcal{N}$ is TTR with $\xi'$ such that $\comm{\xi'}{\xi}=0$, if a prior state $\gamma$ satisfies $\comm{\gamma}{Z}=0$ such that $\gamma$ is a steady state of $\mathcal{N}$. \label{exm:CNOT}
\end{exm}
This example appeared in Theorem 4 in Ref.~\cite{aw2024role} as an example of non-product preserving but tabletop time-reversible dynamics. The tabletop reversibility of $\mathcal{N}$ can be verified explicitly by using \cref{rst:exactTTR}. Let $\gamma,\xi$ and $\xi'$ be in the same basis $\{\ket{i}\}_i$. Then, we get
\begin{align}
\begin{aligned}
    \Phi_{mj\leftarrow nk}\equiv& \bra{m}_{\texttt{S}}\ten\bra*{j}_{\texttt{E}}U\ket{n}_{\texttt{S}}\ten\ket{k}_{\texttt{E}} \\
    =& \left(\bra{m}_{\texttt{S}}\ten\bra*{j}_{\texttt{E}}\right)\Big(\ket{n}_{\texttt{S}}\ten X^n \ket{k}_{\texttt{E}}\Big)\\
    =& \delta_{m,n} \braket*{j}{k\oplus n} \\
    =& \delta_{m,n}\delta_{j,k\oplus n},
\end{aligned}
\end{align}
where $\delta_{m,n}$ represents the Kronecker delta and $k\oplus n$ represents $k+n \mod 2$. Then it is not difficult to see that the exact TTR condition \eqref{eq:iff} can be satisfied via the TTR matrix $M$: having computed the $16 \cross 2$ TTR matrix as  $(M)_{m1m2n1n2,j}=\delta_{m1,n1}\delta_{m2,n2}\delta_{n1,n1}$ and the $16$-dimensional output vector $(\Vec y)_{m1m2n1n2} = \delta_{m1,n1}\delta_{m2,n2}\delta_{n1,n2}$, it is immediate that \emph{any} probability $p_j'$'s satisfy the TTR matrix equation
\begin{align}
    \begin{pmatrix}
        1 & 1 \\
        0 & 0 \\
        \vdots&\vdots\\
        0 & 0 \\
        1 & 1
    \end{pmatrix} \,
    \begin{pmatrix}
        p'_0\\
        p'_1
    \end{pmatrix}
    = 
    \begin{pmatrix}
        1 \\
        0 \\
        \vdots\\
        0  \\
        1 
    \end{pmatrix}.
\end{align}

In turn, we study a continuous time dynamics, which motivates us to define the approximate condition of TTR we introduce later.

\begin{exm}[$XX$ Hamiltonian]
    Consider the dynamics defined by $H_{\texttt{tot}}=gX\ten X$, or $gXX$ for short-handed notation, and an arbitrary ancilla $\xi$. The family of maps 
    \begin{align}
        \N_\theta(\bullet) = \Tr_{\texttt{E}}\left(e^{-iXX \theta}(\bullet\ten\xi)e^{iXX \theta}\right)
    \end{align} describing the dynamics of the system after having interacted with the ancilla for a time $\Delta t$ are TTR for every $\theta=g\Delta t$, if the prior state $\gamma$ satisfies $\comm{\gamma}{X}= 0$. In general, $\xi'$ depends on $\theta$; if $\comm{\xi}{X}= 0$ (thermal channel) or $\comm{\xi}{Z}=0$ (dephasing channel), then $\xi'=\xi$ for all $\theta$.
    \label{ex:XX}
\end{exm}

The dynamics under study is
\begin{eqnarray}
\begin{aligned}
    &e^{-iXX \theta}(\gamma \ten \xi) e^{iXX \theta} \\
    =& (\cos{\theta}\one-i\sin{\theta}XX)(\gamma \ten \xi)(\cos{\theta}\one+i\sin{\theta}XX)\\
    =& \cos^2{\theta} (\gamma \ten \xi)  + \sin^2{\theta} (\gamma \ten X \xi X)\\
    &+ i(\cos{\theta}\sin{\theta}) X \gamma\ten \comm{X}{\xi}, 
\end{aligned}
\end{eqnarray}   
where the condition $\comm{\gamma}{X}= 0$ was used in the last line. Three facts follow from this: First, any $\gamma$ that commutes with $X$ is a steady state of $\N_\theta$ for all $\theta$. Second, $U_\theta=e^{-iXX \theta}$ preserves the product $\gamma \ten \xi$ whenever $\theta=\frac{N}{2}\pi$ for any integer $N$. Third, if $\comm{\xi}{X}= 0$, $U_\theta=e^{-iXX \theta}$ preserves the product $\gamma \ten \xi$ for any $\theta$, as expected from a thermal operation.

Let us look next at the case $\comm{\xi}{Z}=0$, which defines dephasing channels $\N_\theta(\rho) = \cos^2\theta \rho + \sin^2\theta X \rho X$. We first compute the transition matrix $\Phi_{mj\leftarrow nk}$ for $\xi'=\xi$:
\begin{align}
\begin{aligned}
    &\Phi_{mj\leftarrow nk}\\
    \equiv& 
    \bra{\lambda_m}_{\texttt{S}}\bra*{e'_j}_{\texttt{E}} U_{\theta} \ket*{\lambda_n}_{\texttt{S}}\ket*{e_k}_{\texttt{E}} \\
    =& \bra{\lambda_m}_{\texttt{S}}\bra*{j}_{\texttt{E}}(\cos{\theta}\one-i\sin{\theta}XX) \ket*{\lambda_n}_{\texttt{S}}\ket*{k}_{\texttt{E}}\\
    =& \delta_{m,n}\left( \cos\theta \delta_{j,k} -i \sin\theta \delta_{j,k\oplus 1} \right).
\end{aligned}
\end{align}
Having this, the TTR condition is given by 
\begin{align}
\begin{aligned}
     0&=\sum_{jk} \left(p'_j\sqrt{r'_{m_1}r'_{m_2}}-p_k\sqrt{r_{n_1}r_{n_2}}\right)\Phi_{m_1 j\leftarrow n_1k}\Phi^*_{m_2 j\leftarrow n_2k}\\
    &= \sum_{jk} r_{n_1} \left(p_j-p_k\right) \Big( (\cos^2\theta) \delta_{j,k} + (\sin^2\theta) \delta_{j,k\oplus 1} \Big)\\
    &= r_{n_1} \cos^2\theta \Big(\sum_{k} p_{k}-\sum_{k} p_k\Big) +
    r_{n_1} \sin^2\theta  \Big(\sum_{k} p_{k\oplus 1}-\sum_{k} p_k\Big)\,.
\end{aligned}
\end{align}

Finally, we show that $\N_\theta(\bullet)$ is TTR for all $\theta$ whatever the choice of $\xi$, although in general $\xi'$ depends on $\theta$ as $\xi'=u_\theta \xi u_\theta^\dagger$ with $u_\theta=e^{-i\theta X}=\cos\theta \one-i\sin\theta X$. Notice that this implies $\Tr\left[\xi'X\right]=\Tr\left[\xi X\right]$.

The general proof is given in \cref{app:example}; here we show the case $\theta=\frac{\pi}{4}$. First, let us compute the transition matrix $\Phi_{mj\leftarrow nk}$:
\begin{align}
\begin{aligned}
    &\Phi_{mj\leftarrow nk}\\
    \equiv& 
    \bra{\lambda_m}_{\texttt{S}}\bra*{e'_j}_{\texttt{E}} U_{\theta} \ket*{\lambda_n}_{\texttt{S}}\ket*{e_k}_{\texttt{E}} \\
    =& \bra{\lambda_m}_{\texttt{S}}\bra*{e'_j}_{\texttt{E}}(\cos{\theta}\one-i\sin{\theta}XX) \ket*{\lambda_n}_{\texttt{S}}\ket*{e_k}_{\texttt{E}}\\
    =& \bra{\lambda_m}\ket*{\lambda_n}\bra*{e'_j}\Big(\underbrace{\frac{1}{\sqrt{2}}\ket*{e_k}+i\frac{(-1)^{n+1}}{\sqrt{2}} X\ket*{e_{k}}}_{=:\ket*{f_{n,k}}}\Big).
\end{aligned}
\end{align}
Let us then choose $\xi'$ such that $\Tr\left[\xi'X\right]=\Tr\left[\xi X\right]$ and its eigenbasis is given by $\ket*{e'_{j}}:=\frac{1}{\sqrt{2}}\left(\ket*{e_j}- i X\ket*{e_{j}}\right)$. Then we compute $\bra*{e_j'}\ket*{f_{n,k}}$ as follows,
\begin{align}
    \bra*{e_j'}\ket*{f_{n,k}} = 
    \begin{cases}
        \delta_{j,k}, \quad &\text{if} \,\, n=0\\
        i\bra*{e_j}X\ket*{e_k}, \quad &\text{if} \,\,n=1.
    \end{cases} 
\end{align}
Using $\theta=\frac{1}{4}\pi$, this leads to the transition matrix given by 
\begin{align}
    \Phi_{mj\leftarrow nk} = 
    \begin{cases}
        \delta_{m,n}\delta_{j,k}, \quad &\text{if} \,\, n=0\\
        i\delta_{m,n}\bra*{e_j}X\ket*{e_k}, \quad &\text{if} \,\,n=1.
    \end{cases} 
\end{align}
Finally, we verify the TTR condition in each of the four cases $(n_1,n_2)\in\{0,1\}^2$:

\noindent (a) when $n_1=n_2=0$, the TTR condition is given by 
\begin{align}
    \begin{aligned}
        &\sum_{jk} \left(p'_j\sqrt{r'_{m_1}r'_{m_2}}-p_k\sqrt{r_{n_1}r_{n_2}}\right)\Phi_{m_1 j\leftarrow n_1k}\Phi^*_{m_2 j\leftarrow n_2k}= r_{n_1} \sum_{jk} \left(p'_j-p_k\right)\delta_{j,k}=0.
    \end{aligned}
\end{align}
(b) when $n_1=n_2=1$, we have 
\begin{align}
    \begin{aligned}
        \sum_{jk} \left(p'_j\sqrt{r'_{m_1}r'_{m_2}}-p_k\sqrt{r_{n_1}r_{n_2}}\right)\Phi_{m_1 j\leftarrow n_1k}\Phi^*_{m_2 j\leftarrow n_2k}&= r_{n_1} \sum_{jk}\left(p'_j-p_k\right)\bra*{e_j}X\ket*{e_k}\bra*{e_k}X\ket*{e_j}\\
        &= r_{n_1} \Big(\sum_{j} p'_j -
        \sum_{k} p_k \Big)
        =0.
    \end{aligned}
\end{align}
(c) when $n_1=1,n_2=0$, we have 
\begin{align}
    \begin{aligned}
        \sum_{jk} \left(p'_j\sqrt{r'_{m_1}r'_{m_2}}-p_k\sqrt{r_{n_1}r_{n_2}}\right)\Phi_{m_1 j\leftarrow n_1k}\Phi^*_{m_2 j\leftarrow n_2k} &=  i\sqrt{r_{n_1} r_{n_2}} \sum_{jk} \left(p'_j-p_k\right) \delta_{j,k} \bra*{e_j}X\ket*{e_k}\\
        &= i\sqrt{r_{n_1} r_{n_2}} \Big(\sum_{j} p'_j  \bra*{e_j}X\ket*{e_j} -
        \sum_{k} p_k  \bra*{e_k}X\ket*{e_k}\Big)\\
        &= i\sqrt{r_{n_1} r_{n_2}}\left(\Tr\left[\xi'X\right]-\Tr\left[\xi X\right]\right)=0.
    \end{aligned}
\end{align}
(d) when $n_1=0,n_2=1$, we arrive at the same conclusion as in the case (c). 

\section{Approximate TTR}

\begin{figure}[ht]
    \centering
    \includegraphics[width=0.4\linewidth]{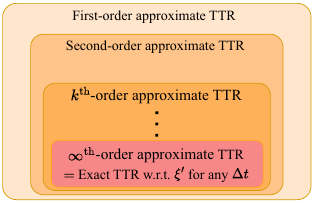}
    \caption{{\bf Approximate Tabletop Time-reversibility.}}
    \label{fig:attr}
\end{figure}

So far, TTR has been defined for a given channel. Example \ref{ex:XX} showed that there are cases in which TTR holds for a whole family of channels indexed by the interaction time, but usually at the price of having to choose $\xi'$ differently for each time. We now switch the focus and search for situations in which TTR holds for $\xi$ independent of the interaction time, at least approximately.

We first define the $k^\texttt{th}$ order approximations. Then use this knowledge to discuss TTR for a random-time collision model, and finish with the application to an excitation-exchange Hamiltonian.

Consider a dynamics $\N_{\Delta t}$ for a short time $\Delta t$. A direct approach to get approximate conditions of TTR is to expand $\hat{{\cal N}}^{(\texttt{P})}_{\Delta t,\gamma}$ and $\hat{{\cal N}}^{(\texttt{T})}_{\Delta t,\xi'}$ in powers of $\Delta t$ up to some desired order and match terms order by order (sketched in \cref{fig:attr}): 
\begin{dfn}
    Given $(H_{\texttt{tot}},\xi)$ that define the channel ${\cal N}_{\Delta t}$ for any evolution time $\Delta t$, we consider the expansions of $\hat{{\cal N}}^{(\texttt{P})}_{\Delta t,\gamma}$ and $\hat{{\cal N}}^{(\texttt{T})}_{\Delta t,\xi'}$ as a power series in $\Delta t$. We say that $(H_{\texttt{tot}},\xi)$ is \emph{$k^{\texttt{th}}$ order approximate tabletop time-reversible for a prior $\gamma$ with regard to $\xi'$} if the $n^{\texttt{th}}$ order terms of the two series are the same for all $n\le k$. 
\end{dfn}

As desired, the condition of approximate TTR does not depend on $\Delta t$; neither will $\xi'$. When $k=0$, the condition is trivial. The limit $k \to \infty$ corresponds to the case of exact TTR with regard to \emph{the same} $\xi'$ for any $\Delta t$. This is a strong condition, but we know that thermal operations satisfy it; we leave the question open as to whether there exist non-thermal operations that are $\infty^{\texttt{th}}$-order approximate TTR.

In what follows, we will present first two non-trivial approximation of TTR conditions. For convenience $X_{\texttt{S}}$ and $X_{\texttt{E}}$ will be used as shorthand notations for $X_{\texttt{S}} \ten \one_{\texttt{E}}$ and $\one_{\texttt{S}} \ten X_{\texttt{E}}$, respectively. Calligraphic fonts are used only for channels and superoperators.
We will also adopt the natural unit $\hbar:=1$, and often suppress the subscript $t$ from states when it is clear that we are concerned with a single dynamics. 

\subsection{First-order approximation}

Let us begin with the first-order approximation of the TTR condition:

\begin{thm}[First-order approximation]
\label{rst:firstTTR}
$(H_{\texttt{tot}},\xi)$ is first-order approximate tabletop time-reversible for a prior $\gamma$ with regard to $\xi'$ if it holds that 
\begin{align}
\begin{aligned}
    &\Tr_{\texttt{E}} \left( \left( \gamma_{\texttt{S}}^{\frac{1}{2}} H_{\texttt{tot}} \gamma_{\texttt{S}}^{-\frac{1}{2}}\right) \xi_{\texttt{E}} \right) - iA\gamma_\texttt{S}^{-\frac{1}{2}} = \Tr_{\texttt{E}} \left( H_{\texttt{tot}} \xi'_{\texttt{E}} \right)\\
    &\Tr_{\texttt{E}} \left( \left( \gamma_{\texttt{S}}^{-\frac{1}{2}} H_{\texttt{tot}} \gamma_{\texttt{S}}^{\frac{1}{2}}\right)  \xi_{\texttt{E}} \right) + i\gamma_\texttt{S}^{-\frac{1}{2}}A = \Tr_{\texttt{E}} \left( H_{\texttt{tot}} \xi'_{\texttt{E}} \right),
\end{aligned}
\label{eq:firstTTR}
\end{align}
where we define $A \equiv -i\int_0^\infty e^{-x \sqrt{\gamma}} \comm{H}{\gamma} e^{-x \sqrt{\gamma}} dx $ and $H\equiv \Tr_{\texttt{E}} \left( H_{\texttt{tot}}\xi_{\texttt{E}} \right)$. When $\gamma$ is a steady state or the maximally mixed state, $A$ vanishes.
\end{thm}

\begin{proof}
    Here, we will give a rough sketch of proof when $\gamma$ is a steady state (see \cref{app:first_order} for the complete proof).
    
    We start with the Kraus representation of $\N_{\Delta t}$:
    \begin{align*}
        \N_{\Delta t}(\rho) = \sum_{j,k} p_k \bra*{e_j}_{\texttt{E}} e^{-i H_{\texttt{tot}}\Delta t}\ket*{e_k}_{\texttt{E}}  \rho  \bra*{e_k}_{\texttt{E}} e^{i H_{\texttt{tot}}\Delta t}\ket*{e_j}_{\texttt{E}}. 
    \end{align*}
    Accordingly, the Petz recovery map for prior being a steady state can be expressed as 
    \begin{align}
    \begin{aligned}
        \hat{{\cal N}}^{(\texttt{P})}_{\Delta t,\gamma}(\rho)&=\sum_{j,k} p_k \gamma^{\frac{1}{2}} \!\bra*{e_k}_{\texttt{E}} e^{i H_{\texttt{tot}}\Delta t}\ket*{e_j}_{\texttt{E}} \!\gamma^{-\frac{1}{2}} \rho \gamma^{-\frac{1}{2}} \!\bra*{e_j}_{\texttt{E}} e^{-i H_{\texttt{tot}}\Delta t}\ket*{e_k}_{\texttt{E}} \! \gamma^{\frac{1}{2}}\\
        &=\sum_{j,k} p_k \gamma^{\frac{1}{2}} \!\bra*{e_k}_{\texttt{E}}(\one+i\Delta t H_{\texttt{tot}})\ket*{e_j}_{\texttt{E}} \gamma^{-\frac{1}{2}} \rho \gamma^{-\frac{1}{2}} (\one-i\Delta t H_{\texttt{tot}}) \ket*{e_k}_{\texttt{E}} \! \gamma^{\frac{1}{2}}\\
        &=\rho+i\Delta t\Tr_{\texttt{E}} \left( (\gamma^{\frac{1}{2}}\ten \one ) H_{\texttt{tot}} ( \gamma^{-\frac{1}{2}} \ten \xi) \right)\rho -i\Delta t\rho\Tr_{\texttt{E}} \left( (\gamma^{-\frac{1}{2}}\ten \one ) H_{\texttt{tot}} ( \gamma^{\frac{1}{2}} \ten \xi)\right)+\order{(\Delta t)^2}.
    \end{aligned} 
    \end{align}
    Similarly, we get 
    \begin{align}
    \begin{aligned}
        \hat{\N}_{\Delta t,\xi'}^{(\texttt{T})}(\rho) &= \sum_{j,k} p'_j  \bra*{e_k'}_{\texttt{E}} e^{i H_{\texttt{tot}}\Delta t}\ket*{e_j'}_{\texttt{E}} \rho  \bra*{e_j'}_{\texttt{E}} e^{-i H_{\texttt{tot}}\Delta t}\ket*{e_k'}_{\texttt{E}}\\
        &= \rho + i\Delta{t}\sum_{j} p'_j  \bra*{e_j'}_{\texttt{E}}  H_{\texttt{tot}} \ket*{e_j'}_{\texttt{E}} \rho -i\Delta{t}\rho\sum_{j} p'_j  \bra*{e_j'}_{\texttt{E}} H_{\texttt{tot}} \ket*{e_j'}_{\texttt{E}}  + \order{(\Delta t)^2}\\
        &=\rho+i(\Delta t)\Tr_{\texttt{E}} \left( H_{\texttt{tot}} ( \one \otimes \xi') \right)\rho -i(\Delta t)\rho \Tr_{\texttt{E}} \left( H_{\texttt{tot}} (\one \otimes \xi') \right) + \order{(\Delta t)^2}. \nonumber
    \end{aligned}
    \end{align}
\end{proof}

From \cref{rst:firstTTR}, we can make some useful observations. For the case where $\gamma$ is the maximally mixed state, conditions \eqref{eq:firstTTR} reduce to 
\begin{align}
    \Tr_{\texttt{E}} \left( H_{\texttt{tot}} \xi_{\texttt{E}} \right) = \Tr_{\texttt{E}} \left( H_{\texttt{tot}} \xi'_{\texttt{E}} \right).
\end{align}
In particular, any dynamics is first-order approximate TTR with regard to $\xi'=\xi$ when $\gamma=\one/d$. As such, one can approximately implement a Petz recovery map for $\gamma=\one/d$ by running $U_{\Delta t}=e^{-iH_{\texttt{tot}}\Delta t}$ backwards, i.e., $U_{(-\Delta t)}=U_{\Delta t}^\dagger$. If the maximally mixed state is a steady state of the original dynamics, sequential implementation of Petz recovery maps via TTR is possible without having to change $\xi'$. This is similar to how thermal operations are TTR by running the same operation backwards. The same conclusion can be derived when $\gamma$ is a steady state and satisfies $[\gamma_{\texttt{S}}, H_{\texttt{tot}}] = 0$.

As an example, recall \cref{ex:XX} where we considered $H_{\texttt{tot}}=gXX$ and a steady state $\gamma$. In this case, $\gamma$ satisfies $[\gamma_{\texttt{S}}, H_{\texttt{tot}}] = 0$. The dynamics under consideration is then TTR for any evolution time, but $\xi'$ depends on an evolution time unless $\xi$ is such that the dynamics becomes the trivial case of thermal operations or dephasing channels. Nonetheless, \cref{rst:firstTTR} says that the dynamics can be TTR by using the same ancilla $\xi'=\xi$ regardless of an evolution time $\Delta t$, while sacrificing accuracy in implementing the Petz recovery map up to an error $\order{(\Delta t)^2}$. Besides, since the effect of a unitary for the interaction is periodic with respect to an evolution time, there exists $\Delta \hat{t} >0$ such that $U_{\Delta \hat{t}}=U_{(-\Delta t)}$.

Moreover, \cref{rst:firstTTR} holds an useful implication in quantum error correction. A Petz recovery map for prior being the maximally mixed state on a code space (also known as `transpose' channel or code-specific Petz recovery map) is often of interest, and our results imply that such Petz recovery maps are approximately implementable by running the same environment backwards with an error of the second order of the evolution time.

\subsection{Second-order approximation}

Similarly, we can find the conditions for the second-order approximate TTR by the second-order matching of the two maps. The final form of the second-order terms of a Petz recovery map is complicated in general: it can be found in \cref{app:second_order}, together with its derivation. Here we present the conditions for the second-order approximate TTR in the special case of a steady state $\gamma$ satisfying $[\gamma_{\texttt{S}}, H_{\texttt{tot}}] = 0$:

\begin{thm}[Second-order approximation]
\label{rst:secondTTR}
If a prior state $\gamma_{\texttt{S}}$ is a steady state and $[\gamma_{\texttt{S}}, H_{\texttt{tot}}] = 0$, then $(H_{\texttt{tot}},\xi_{\texttt{E}})$ is second-order approximate tabletop time-reversible with regard to $\xi'_{\texttt{E}}$ if it holds that 
\begin{align}
    \Tr_{\texttt{E}} \left( H_{\texttt{tot}} \xi_{\texttt{E}} \right) = \Tr_{\texttt{E}} \left( H_{\texttt{tot}} \xi'_{\texttt{E}} \right)    
\end{align}
and 
\begin{align}
\begin{aligned}
    &\Tr_{\texttt{E}}\left( H_{\texttt{tot}}\rho_{\texttt{S}}H_{\texttt{tot}} \xi_{\texttt{E}} \right)-\frac{1}{2}\left(\Tr_{\texttt{E}}(  H_{\texttt{tot}}^2 \xi_{\texttt{E}} )\rho_{\texttt{S}}+\rho_{\texttt{S}}\Tr_{\texttt{E}}(  H_{\texttt{tot}}^2 \xi_{\texttt{E}} )\right)\\
    &= \Tr_{\texttt{E}}(  H_{\texttt{tot}}(\rho_{\texttt{S}}\ten \xi'_{\texttt{E}} )H_{\texttt{tot}} ) -\frac{1}{2}\left(\Tr_{\texttt{E}}(  H_{\texttt{tot}}^2 \xi'_{\texttt{E}} )\rho_{\texttt{S}}+\rho_{\texttt{S}}\Tr_{\texttt{E}}(  H_{\texttt{tot}}^2 \xi'_{\texttt{E}} )\right),
\end{aligned}
\end{align}
for any $\rho_\texttt{S}$.
\end{thm}

If the Hamiltonian satisfies $H_{\texttt{tot}}^2=\one$, the second-order approximate TTR condition can be further simplified to 
\begin{align}
\begin{aligned}
    &\Tr_{\texttt{E}}\left( H_{\texttt{tot}}\rho_{\texttt{S}}H_{\texttt{tot}} \xi_{\texttt{E}} \right)= \Tr_{\texttt{E}}(  H_{\texttt{tot}}(\rho_{\texttt{S}}\ten \xi'_{\texttt{E}} )H_{\texttt{tot}} ), \quad \forall \rho_\texttt{S}.
\end{aligned}
\end{align}
An example of such a Hamiltonian includes $XX$ Hamiltonian in \cref{ex:XX}.

So far, we have considered dynamics for a small evolution time $\Delta t$. Now imagine a dynamics $\N_{T}$ for a time $T$ which is not sufficiently small. One can attempt to recover the noisy dynamics $\N_{T}$ by executing Petz recovery maps that recover dynamics $\N_{\Delta t}$ for a smaller time $\Delta t\equiv \frac{T}{N}$. 
If the Lindbladian~\cite{lindblad1976lindblad,gorini1976lindblad,manzano2020short,stefanini2025lindblad} of the dynamics is known, the entire dynamics can be written as
\begin{align}
    \N_{T} = \circ_{n=1}^{N} \N_{\Delta t} =  e^{{\Delta t} \mathcal{L}} \circ \cdots \circ e^{{\Delta t} \mathcal{L}},
\end{align}
where $\mathcal{L}$ is the Lindbladian of the dynamics of interest. Then, one can apply conditions of approximate TTR to each dynamics $\mathcal{N}_{\Delta t}$ up to some desired order. For instance, under the conditions of the second order approximate TTR, Petz recovery maps can be realized by implementing tabletop reverse maps sequentially, thus \emph{sequential TTR} or \emph{composable TTR} \cite{aw2024role} with an error $\order{\frac{T^3}{N^2}}$, namely, for any $\rho$
\begin{align*}
    \frac{1}{2}\norm{\left(\circ_{n=1}^{N} \hat{\N}^{(\texttt{P})}_{\Delta t,\gamma_{n-1}} 
    - 
    \circ_{n=1}^{N} \hat{\N}^{(\texttt{T})}_{\Delta t,\xi'_{n-1}}\right)(\rho)
    }_1 \leq \order{N(\Delta t)^3} = \order{\frac{T^3}{N^2}},
\end{align*}
where $\gamma_n$ and $\xi'_n$ represent a prior state of a Petz recovery map and an ancilla state of the corresponding tabletop reverse map at each time step $t_n=n\Delta t$.

\subsection{Lindbladian from a random-time collision model}

\begin{figure}[h]
    \centering
    \includegraphics[width=0.65\linewidth]{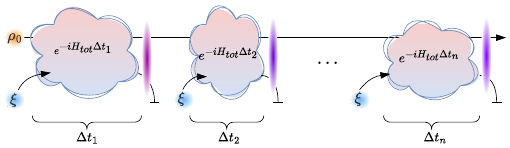}
    \caption{{\bf Random-time Collision Model.}}
    \label{fig:collision}
\end{figure}

The typical derivation of the Lindbladian, however, uses a continuous-time formalism from the outset, although it is implicitly acknowledged that the Born-Markov approximation cannot hold for arbitrarily short times (there must be some time of real interaction, during which the evolution is non-Markovian). Besides, the final formulas are hard to cast as an explicit dependence from the system-bath interaction. In what follows, we consider a collision model~\cite{ scarani2002thermalizing,alicki2007collision,ciccarello2022collision,lacroix2025making}, where we do not rely on the exact Lindbladian of a given dynamics in order to derive conditions of such a sequential TTR.

Consider a random-time collision model, i.e.~a sequential interaction of the system with individual components of the bath (c.f., Fig.~\ref{fig:collision}), where the collision time may vary. Let $gH_{\texttt{tot}}$ be the total Hamiltonian, where we singled out the relevant energy scale $g$ for the problem under study. Without loss of generality, we set $H_{\texttt{tot}} =H_{\texttt{S}}+H_{\texttt{E}}+H_{\texttt{I}}$. A single collision with $\xi$ is given by 
\begin{align}
\begin{aligned}
    \rho_{n} \to \rho_{n+1} &= \Tr_{\texttt{E}}\left(e^{-i gH_{\texttt{tot}}{\Delta t} }(\rho_n \ten \xi_{\texttt{E}})e^{i gH_{\texttt{tot}}{\Delta t}} \right),\\
    &\approx \rho_n + g\Delta t \mathcal{H}\rho_n + (g\Delta t)^2 \sum_{j,k}p_k \mathcal{D}[L_{jk}]\rho_n, \label{eq:collision}
\end{aligned}
\end{align}
having defined the superoperators 
\begin{align}
    \mathcal{H}\rho &\equiv -i\comm{H_{\texttt{S}}+H_{\texttt{Ls}}(\xi)}{\rho}, \\ \mathcal{D}[L_{jk}]\rho &\equiv L_{jk} \rho L_{jk}^\dagger - \frac{1}{2}\acomm{L_{jk}^\dagger L_{jk}}{\rho},
\end{align}
with 
\begin{align}
    H_{\texttt{Ls}}(\xi)&\equiv\Tr_{\texttt{E}} \left((\one_{\texttt{S}}\ten \xi_{\texttt{E}}) H_{\texttt{I}}\right), \\
    L_{jk} &\equiv \bra*{f_j}_{\texttt{E}}{H_{\texttt{tot}}}\ket{e_k}_{\texttt{E}},
\end{align}
for any orthonormal basis $\{\ket*{f_j}\}_j$. Similar to the exact TTR, we will choose $\ket*{f_j}:=\ket*{e'_j}$ the eigenvectors of $\xi'$. As such, Eq.~\eqref{eq:collision} motivates the Lindbladian 
\begin{align}
    \mathcal{L}= g\mathcal{H} + \Gamma \sum_{j,k} p_k \mathcal{D}[L_{jk}],
\end{align}
with the rate $\Gamma$ of the event of a collision. When collision times are finite and $\Gamma$ is sufficiently small such that $\Gamma\sim g^2 \ev{\Delta t}$, this Lindbladian captures a coarse-grained behavior of the scattering process $(H_{\texttt{tot}},\xi)$, with the superoperators $\mathcal{H}$ and  $\Gamma p_k \mathcal{D}[L_{jk}]$. The superoperator $\mathcal{H}$ accounts for a closed dynamics governed by the system Hamiltonian moderated by a so-called Lamb shift Hamiltonian $H_{\texttt{Ls}}(\xi)$, which describes a shift of the system's energy levels due to its interaction with the ancilla. The superoperator $\Gamma p_k\mathcal{D}[L_{jk}]$, called dissipator, accounts for a coarse-grained dissipation in the direction of jump operators $L_{jk}$ with its weight $\Gamma p_k$. 

Under the same collision model, the tabletop reverse Lindbladian of the reverse collisional dynamics with ancilla $\xi'$ is given by 
\begin{align}
\hat{\mathcal{L}}^{(\texttt{T})}_{\xi'}=g\hat{\mathcal{H}}^{(\texttt{T})}_{\xi'} + \Gamma \sum_{j,k} p'_j \mathcal{D}[\hat{L}_{jk}^{(\texttt{T})}]
\end{align}
with a Hamiltonian and jump operators 
\begin{align}
     \hat{\mathcal{H}}_{\xi'}^{(\texttt{T})} \rho =& -i \comm{ -H_{\texttt{S}}-H_{\texttt{Ls}}(\xi')}{\rho} \\
     \hat{L}^{(\texttt{T})}_{jk} =& \,L_{jk}^\dagger, 
\end{align}   
describes tabletop reverse collisions given by $(-H_\texttt{tot},\xi')$.

Following Ref.~\cite{kwon2022reversing}, for a prior state $\gamma_t$  
whose spectral decomposition $\gamma_t\stackrel{s.d}{=}\sum_{\lambda_t} \lambda_t \op{\lambda_t}$ at time $t$, we compute a Lindbladian
\begin{align}
    \hat{\mathcal{L}}^{(\texttt{P})}_{\gamma_t} \equiv g\hat{\mathcal{H}}^{(\texttt{P})}_{\gamma_t} + \Gamma\sum_{j,k} p_k\mathcal{D}[\hat{L}^{(\texttt{P})}_{jk}]
\end{align}
with a Hamiltonian and jump operators 
\begin{align}
    \hat{\mathcal{H}}_{\gamma_t}^{(\texttt{P})}\rho&\equiv -i\comm{-H_{\texttt{S}} -H_{\texttt{Ls}}(\xi) + H_{\texttt{C}}}{\rho}\\
    \hat{L}^{(\texttt{P})}_{jk}&=\hat{L}^{(\texttt{P})}_{jk}(\gamma_t)\equiv\gamma_t^{1/2}L_{jk}^\dagger\gamma_t^{-1/2}, 
\end{align}
defining the correction Hamiltonian $H_{\texttt{C}}=H_{\texttt{C},\gamma_t}$ as
\begin{eqnarray}
\begin{aligned}
    \!\!\!\!\!H_{\texttt{C},\gamma_t}\equiv\frac{1}{2i}\sum_{\lambda_t,\lambda_t'}\left(\frac{\sqrt{\lambda_t}-\sqrt{\lambda_t'}}{\sqrt{\lambda_t}+\sqrt{\lambda_t'}}\right)\bra{\lambda_t}M(\gamma_t) \ket*{\lambda_t'} \op{\lambda_t}{\lambda_t'},\label{eq:hc}
\end{aligned}
\end{eqnarray}
with $M(\gamma_t)\equiv \Gamma\sum_{jk}p_k (L_{jk}\+L_{jk}+ \gamma_t^{-1/2}L_{jk}\gamma_t L_{jk}\+\gamma_t^{-1/2})$. The Lindbladian $\hat{{\cal L}}^{(\texttt{P})}_{\gamma_t}$ is not the exact Lindbladian of the Petz recovery map of the forward collisional dynamics unless $\gamma_t$ is a steady state, but its Lindblad dynamics is sufficiently close to the Petz recovery map. The following lemma shows this (see \cref{app:proof_lem} for the proof):

\begin{lem} \label{lem:Petz_Lindblad}
Consider a collisional dynamics $\N_{\Delta t}:=e^{\Delta t \mathcal{L}}$ for an evolution time $\Delta t$. Let $\hat{{\cal N}}^{(\texttt{P})}_{\Delta t, \gamma_t}$ be the Petz recovery map of the dynamics $\N_{\Delta t}$ for a prior $\gamma_t$ at a time $t$. The Lindblad dynamics $e^{ \Delta t \hat{{\cal L}}^{(\texttt{P})}_{\gamma_t}}$ is close to the Petz recovery map $\hat{{\cal N}}^{(\texttt{P})}_{\Delta t, \gamma_t}$ up to an error ${\cal O}\left( (\Delta t)^2 \right)$, i.e., 
    \begin{equation}
        \frac{1}{2} \left \| \hat{{\cal N}}^{(\texttt{P})}_{\Delta t, \gamma_t} - e^{ \Delta t \hat{{\cal L}}^{(\texttt{P})}_{\gamma_t}} \right\|_\diamond = {\cal O}\left( (\Delta t)^2 \right),    
    \end{equation}
    where $\norm{\bullet}_\diamond$ is the diamond norm~\cite{aharonov1998diamond}.
\end{lem}

This immediately leads to the Lindbladian conditions of approximate TTR (see \cref{app:Lindblad_proof} for the proof): 
 
\begin{thm}\label{rst:Lindblad}
    Consider a collisional dynamics $\N_{\Delta t}:=e^{\Delta t \mathcal{L}}$ and a collisional tabletop reverse dynamics $e^{\Delta t \hat{\mathcal{L}}^{(\texttt{T})}_{\xi'}}$. A Petz recovery map $\hat{{\cal N}}^{(\texttt{P})}_{\Delta t,\gamma_t}$ of the dynamics $\N_{\Delta t}$ can be approximately implemented by $e^{\Delta t \hat{\mathcal{L}}^{(\texttt{T})}_{\xi'}}$ up to an error $\order{(\Delta t)^2}$, i.e., 
    \begin{align}
        \frac{1}{2}\norm{
        \hat{{\cal N}}^{(\texttt{P})}_{\Delta t,\gamma_t} - e^{\Delta t \hat{\mathcal{L}}^{(\texttt{T})}_{\xi'}}
        }_{\diamond}=\order{(\Delta t)^2}, 
    \end{align}
    if $\hat{\mathcal{L}}^{(\texttt{P})}_{\gamma_t}=\hat{\mathcal{L}}^{(\texttt{T})}_{\xi'} + \order{\Delta t}$ or equivalently, 
    \begin{align}
    \begin{aligned}
        &
    H_{\texttt{Ls}}(\xi) - H_{\texttt{C}} = H_{\texttt{Ls}}(\xi') + \alpha\one + \order{\Delta t} 
    \\
    &\sum_{jk}p_k\mathcal{D}[\gamma_t^{1/2}\bra*{e_k}_{\texttt{E}}H_{\texttt{tot}}\ket*{e'_j}_{\texttt{E}}\gamma_t^{-1/2}]  =\sum_{jk}p'_j\mathcal{D}[\bra*{e_k}_{\texttt{E}}H_{\texttt{tot}}\ket*{e'_j}_{\texttt{E}}] + \order{\Delta t}, \label{eq:TTRcondition_HD}
    \end{aligned}
\end{align}
for some number $\alpha \in \mathbb{R}$.
Moreover, the complex term of the correction Hamiltonian $H_{\texttt{C}}=H_{\texttt{C},\gamma_t}$ vanishes if 
    \begin{eqnarray}
    \begin{aligned}
        \gamma_t^{1/2}\bra*{e_k}_{\texttt{E}}{H_{\texttt{tot}}}\ket*{e'_j}_{\texttt{E}}\gamma_t^{-1/2}=c_{jk}\bra*{e_k}_{\texttt{E}}{H_{\texttt{tot}}}\ket*{e'_j}_{\texttt{E}},
    \end{aligned}
    \end{eqnarray}
    for some number $c_{jk}$.
\end{thm}

Under this collision model, one can proceed to an approximate implementation of Petz recovery maps through a sequential TTR when the evolution time $T$ of the noisy dynamics is not sufficiently small:

\begin{cor}\label{rst:sequentialTTR}
    The concatenation of Petz recovery maps for each collision and that of tabletop reverse maps satisfy 
    \begin{align}\label{eq:sequential_dn}
    \frac{1}{2}\norm{\circ_{n=1}^{N} \hat{\N}^{(\texttt{P})}_{\Delta t,\gamma_{n-1}} - \circ_{n=1}^{N} e^{\Delta t \hat{\mathcal{L}}^{(\texttt{T})}_{\xi'_{n-1}}}}_{\diamond} = \order{\frac{T^2}{N}}
\end{align}
if it holds that
\begin{eqnarray}
        &H_{\texttt{Ls}}(\xi) = H_{\texttt{Ls}}(\xi'_n) + \alpha \one + \order{\Delta t} \qand \label{eq:sequentialTTR_H} \\
        &\sqrt{p_k} \gamma_n^{1/2}\bra*{e_k}_{\texttt{E}}{H_{\texttt{tot}}}\ket*{e'_j(n)}_{\texttt{E}}\gamma_n^{-1/2} = \sqrt{p'_j(t_n)}\bra*{e_k}_{\texttt{E}}{H_{\texttt{tot}}}\ket*{e'_j(n)}_{\texttt{E}} + \order{\Delta t},
        \label{eq:sequentialTTR_jump}
    \end{eqnarray}
    for some number $\alpha \in \mathbb{R}$, where $\xi'_n \stackrel{s.d.}{=} \sum_{j} p'_j(n) \op*{e'_j(n)}$.
\end{cor}

\begin{proof}
	Following \cref{rst:Lindblad}, if Eq.~\eqref{eq:sequentialTTR_H} and Eq.~\eqref{eq:sequentialTTR_jump} hold for every $n$, the difference between a Petz recovery map and tabletop reverse map for each time intervals scales as 
    \begin{align}
        \frac{1}{2} \left( \hat{\N}^{(\texttt{P})}_{\Delta t,\gamma_{n-1}} -  e^{\Delta t \hat{\mathcal{L}}^{(\texttt{T})}_{\xi'_{n-1}}} \right) (\rho) = \order{(\Delta t)^2}
    \end{align}
    for any $\rho$. Then it immediately follows that 
    \begin{align}
    \begin{aligned}
        &\frac{1}{2} \left( \circ_{n=1}^{N}\hat{\N}^{(\texttt{P})}_{\Delta t,\gamma_{n-1}} -  \circ_{n=1}^{N} e^{\Delta t \hat{\mathcal{L}}^{(\texttt{T})}_{\xi'_{n-1}}} \right) (\rho) \\
        =& \frac{1}{2} \left( \circ_{n=1}^{N}\hat{\N}^{(\texttt{P})}_{\Delta t,\gamma_{n-1}} -  \circ_{n=1}^{N} \left( e^{\Delta t \hat{\mathcal{L}}^{(\texttt{T})}_{\xi'_{n-1}}} - 2{\cal O}((\Delta t)^2) \right) \right) (\rho)\\
        =& \underbrace{{\cal O}((\Delta t)^2)+{\cal O}((\Delta t)^2)+\cdots+{\cal O}((\Delta t)^2)}_{N}.
    \end{aligned}
    \end{align}
    Since $\Delta t = \frac{T}{N}$, we arrive at 
    \begin{align}
        \frac{1}{2}\norm{\circ_{n=1}^{N} \hat{\N}^{(\texttt{P})}_{\Delta t,\gamma_{n-1}} - \circ_{n=1}^{N} e^{\Delta t \hat{\mathcal{L}}^{(\texttt{T})}_{\xi'_{n-1}}}}_{\diamond} = \order{\frac{T^2}{N}}.
    \end{align}
\end{proof}

It is worth noting that $\gamma_{n}$'s have been merely referred to the prior states at times $t_n$. Yet, when it comes to sequential reversal, it is important that the prior states must be chosen such that they follow the trajectory determined by the forward dynamics. To facilitate TTR along the trajectory, one thus needs to 
take into account propagated prior states.

\subsection{Preservation of near optimality}

Our definition of approximate TTR preserves near optimality of the Petz recovery map. For a channel $\mathcal{N}_{\Delta t}$, let $ \hat{\N}^{(\texttt{T})}_{\Delta t} $ be $k^\text{th}$ order approximate tabletop reverse map so that $\norm{\hat{\N}^{(\texttt{P})}_{\Delta t} - \hat{\N}^{(\texttt{T})}_{\Delta t}}_\diamond \le \order{(\Delta t)^{k+1}}$. Since a Petz recovery map is a near optimal recovery with respect to entanglement fidelity (see \cref{eq:near_optimal}), and $\norm{\hat{\N}^{(\texttt{P})}_{\Delta t} - \hat{\N}^{(\texttt{T})}_{\Delta t}}_\diamond \le \order{(\Delta t)^{k+1}}$, the recovery performance of the approximate TTR realization of the Petz recovery map, quantified by the entanglement fidelity, remains sufficiently close to near optimal. Formally, for any $\rho$, we have

\begin{align}
    \abs{F_\texttt{ent}(\rho, \hat{\N}^{(\texttt{T})}_{\Delta t}\circ \mathcal{N}_{\Delta t}) - F_\texttt{ent}(\rho, \hat{\N}^{(\texttt{P})}_{\Delta t}\circ \mathcal{N}_{\Delta t})} & = 
    \abs{\bra*{\Psi_{\rho}} (id \ten \hat{\N}^{(\texttt{T})}_{\Delta t}\circ \mathcal{N}_{\Delta t}-\hat{\N}^{(\texttt{P})}_{\Delta t}\circ \mathcal{N}_{\Delta t})\left(\op*{\Psi_{\rho}}\right) \ket*{\Psi_{\rho}}
    }\\
    &\le \abs{\norm{ (id \ten \hat{\N}^{(\texttt{T})}_{\Delta t}\circ \mathcal{N}_{\Delta t}-\hat{\N}^{(\texttt{P})}_{\Delta t}\circ \mathcal{N}_{\Delta t})\left(\op*{\Psi_{\rho}}\right) }_1}\\
    & \le \norm{\hat{\N}^{(\texttt{P})}_{\Delta t} - \hat{\N}^{(\texttt{T})}_{\Delta t}}_\diamond \le \order{(\Delta t)^{k+1}},
\end{align}
where $\ket{\Psi_\rho}$ is the purification of $\rho$ and $\norm{M}_1$ denotes the trace norm (also known as Schatten 1-norm) of $M$. The first inequality holds due to $\abs{\bra{\Psi_\rho}M\ket{\Psi_\rho}} \le \norm{M}_\infty \le \norm{M}_1$ for the normalized state $\ket*{\Psi_\rho}$. The second inequality follows from the definition of the diamond norm together with $\norm{(id\ten \mathcal{N}_{\Delta t})\op{\Psi_\rho}}_1 = 1$.

\section{Application} \label{sec:application}

Here we use an excitation-exchange Hamiltonian to illustrate how the Lindbladian
TTR conditions can be applied when the forward dynamics is not exactly
tabletop-reversible and the total evolution time $T$ is not sufficiently small. Consider  
\begin{align}
    H_\texttt{tot} = \sigma_+\ten \sigma_- + \sigma_-\ten \sigma_+ \qand \xi= p_0 \op{0} + p_1 \op{1},
\end{align}
where $\sigma_+ = \op{1}{0}$ and $\sigma_- = \op{0}{1}$. The exact TTR condition is overcomplicated in general, so to extract structural
insight we take the prior $\gamma$ and the reverse ancilla $\xi'$ to be diagonal
in the same computational basis $\{\ket{i}\}_i$ as $\xi$,
\begin{align}
    \gamma = r_0 \op{0} + r_1 \op{1} \qand \xi' = p_0' \op{0} + p_1' \op{1}.
\end{align}
When the system and the ancilla share the same dimension, the requirement is
highly stringent, and only a few pairs $(\gamma,\xi')$ satisfy
Eqs.~\eqref{eq:sequentialTTR_H} and~\eqref{eq:sequentialTTR_jump}. The sharing of the same
basis renders Eq.~\eqref{eq:sequentialTTR_H} automatically satisfied, leaving
only Eq.~\eqref{eq:sequentialTTR_jump} to verify. Imposing the exact satisfaction of Eq.~\eqref{eq:sequentialTTR_jump} yields
$p_0 r_1 = p_1' r_0$ and $p_1 r_0 = p_0' r_1$, which collectively imply
$p_0' p_1' = p_0 p_1$. Since these parameters are probabilities, the
normalization conditions constrain the admissible pairs even further,
leaving only two possibilities: the dynamics is TTR for $\gamma = \xi$ with
$\xi' = \xi$, or for $\gamma = \frac{\one}{2}$ with $\xi' = X \xi X$. Relaxing Eq.~\eqref{eq:sequentialTTR_jump} to hold only to $\order{\Delta t}$
admits a few more cases. The lesson is that a higher-dimensional ancilla allows for a rich family of TTR pairs
$(\gamma,\xi')$, part of which takes no
part in the dynamics and serves only as a residual subspace that absorbs the
normalization condition.

In turn, we show explicitly how a higher-dimensional ancilla admits a far larger family of
prior states. Let us take an explicit example of the three-dimensional ancilla state $\xi = \frac{1}{5}\op{0} + \frac{3}{10}\op{1} + \frac{1}{2}\op{2}$, while keeping the noisy system two-dimensional. For a given prior $\gamma$, the
ancilla $\xi'$ defined by
\begin{align}
    p_0' := p_1\frac{r_0}{r_1}, \quad p_1' := p_0\frac{r_1}{r_0}, \qand p_2':= 1- p_0' - p_1' \label{eq:p_prime}
\end{align}
satisfies Eq.~\eqref{eq:sequentialTTR_jump}. Recall that since $\xi'$ is diagonal in $\{\ket{i}\}_i$,
Eq.~\eqref{eq:sequentialTTR_H} holds automatically. Hence, whenever the
resulting $p_j'$ form a valid probability, the pair
$(\gamma,\xi')$ realizes a sequential TTR implementation of the Petz recovery
map. Specifically, any prior state with
$0.1761 \le r_0 \le 0.7573$ gives a valid ancilla state $\xi'$ via Eq.~\eqref{eq:p_prime}, in contrast to the two dimensional ancilla case. Furthermore, the sequential TTR scheme propagates the prior state along the forward trajectory, with a correspondingly updated
ancilla $\xi'_n$ at each step $n$. Starting from a prior state with 
$r_0 = 3/4$, for example, the prior state is propagated toward the steady state, and thus stays within $0.1761 \le r_0 \le 0.7573$ at every step. Hence each updated $\xi'_n$ is a valid
ancilla state, and the sequential TTR is well-defined.

When $\gamma=\frac{2}{5}\op{0} + \frac{3}{5}\op{1}$ such that it is the steady state,
the dynamics is exact TTR with $\xi' = \xi$, trivially by product
preservation. The same is not guaranteed for a non-steady prior state, yet it can
deliver a better recovery, and enlarging the set of admissible priors  only
helps. This is where the approximate TTR conditions come into play. Relaxing the exact TTR condition makes the Petz recovery map `easily'
implementable for a wider range of prior states, at the cost of an approximate
rather than exact implementation, while the recovery performance remains pretty good.

We end this section by demonstrating \cref{rst:sequentialTTR} and showing that a non-steady prior state can outperform the steady-state
choice. Suppose the system is prepared in $\ket{0}$ with probability
$\tfrac{3}{4}$ and in $\ket{1}$ with probability $\tfrac{1}{4}$, so that the
average initial state is $\rho_\texttt{int} = \frac{3}{4}\op{0} + \frac{1}{4}\op{1}$,
and then undergoes the excitation-exchange noise above. \cref{fig:demo} shows that the diamond distance between the sequential tabletop reverse implementation and the exact Petz
recovery map decreases as the number of collisions $N$ increases and quadratically grows with the total evolution time $T$, as proved in~\cref{rst:sequentialTTR}. Furthermore, choosing the prior state to $\gamma = \rho_\texttt{int}$, gives a higher entanglement
fidelity of recovery than choosing $\gamma$ to be the steady state. \cref{fig:ent_fidelity} shows that the
entanglement fidelity grows with $N$, while the non-steady prior state
$\gamma = \rho_\texttt{int}$ retains its advantage over the steady-state choice throughout.

\begin{figure}[ht]

\includegraphics[width=0.48\textwidth]{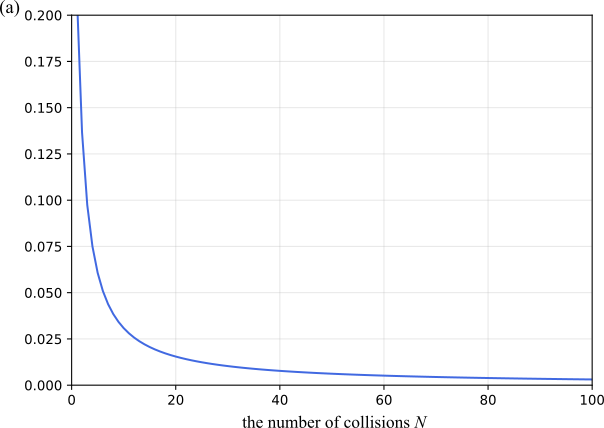}
\includegraphics[width=0.47\textwidth]{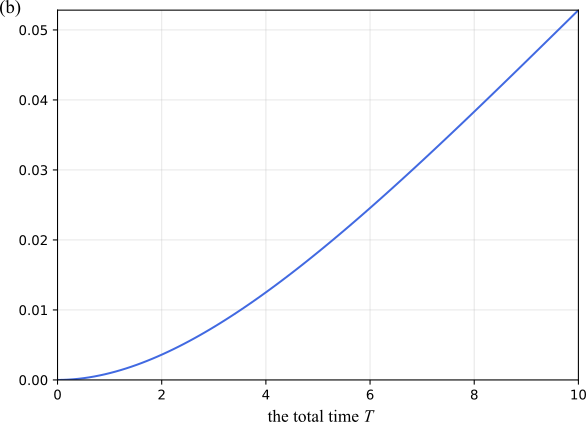}
    
\caption{\textbf{Approximate TTR of the Collision Model for the Excitation-exchange Hamiltonian.} 
The coarse-grained collision rate is $\Gamma = 1$. The left panel (a) shows a graph of the diamond distance between the composition of $N$ Petz recovery maps with propagated prior states and the composition of $N$ tabletop reverse Lindbladian dynamics, i.e., the left-hand side of Eq.~\eqref{eq:sequential_dn}, against the number of collisions $N$, for a fixed $T=3$. The right panel (b) shows a graph of the diamond distance between the two maps as for the first graph, against the total evolution time $T$, for a fixed $N=20$. }
    \label{fig:demo}
\end{figure}

\begin{figure}[ht]
    \centering
    \includegraphics[width=0.45\linewidth]{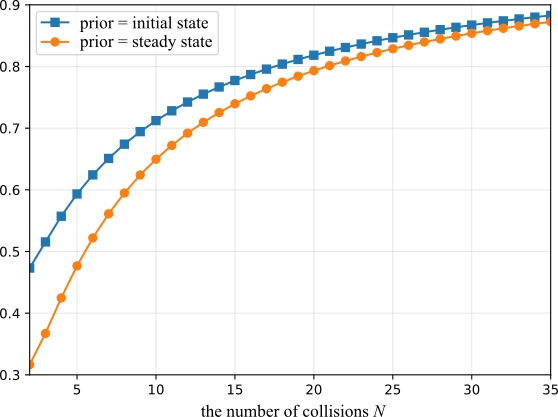}
    \caption{\textbf{Recovery Performance.} The total evolution time is $T = 3$ and the coarse-grained collision rate is $\Gamma = 1$. The figure shows the graph of the entanglement fidelity $F_{\texttt{ent}}(\rho_\texttt{int}, \mathcal{R}\circ \mathcal{N})$ against the number of collisions $N$ for the initial state $\rho_\texttt{int}=\frac{3}{4} \op{0} + \frac{1}{4} \op{1}$, where $\mathcal{R}$ is the composition of the $N$ tabletop reverse maps, i.e., $\exp(\Delta t \hat{\mathcal{L}}^{(\texttt{T})}_{\xi'_{n-1}})$ and $\mathcal{N}$ the forward channel. The blue curve corresponds to the composition of tabletop reverse maps of the Petz recovery maps for the steady-state prior state (thus no update for prior is needed), and the yellow curve to the composition of tabletop reverse maps of the Petz recovery maps when the initial non-steady prior state is $\gamma_1 = \rho_\texttt{int}$ and subsequent $\gamma_n$'s are updated following the dynamics' trajectory.}
    \label{fig:ent_fidelity}
\end{figure}

\section{Discussion}

The Petz recovery map has proven to be a powerful mathematical tool in many applications of quantum information. However, its physical realization has been considered challenging due to its complex form. 
In this work, we study the exact and approximate conditions for tabletop time-reversibility (TTR)~\cite{aw2024role,chen2026ttr_covariant}, under which the Petz recovery map can be implemented using resources similar to, or even identical to, those accessible from the forward dynamics. Thus, TTR-based approaches enable the implementation of a Petz recovery map to be less resource-intensive.

The TTR was first introduced in Ref.~\cite{aw2024role} and its sufficient condition was found in the same paper. Here, we aimed at an analytical characterization of both exact and approximate TTR. First, we derived an explicit condition for exact TTR that disentangles the role of the eigenvalues of the ancilla state $\xi'$ of the tabletop reverse map from that of its eigenbasis, and we illustrated its use with three examples. The matrix equation version of this condition further reveals that a higher-dimensional ancilla can increase the chance of achieving exact TTR. Second, through these examples, we made an observation that an ancilla $\xi'$ for a tabletop reverse map depends on the evolution time in general. This is practically undesired as this would require a sensitive time control for $\xi'$ for TTR. In particular, the first two non-trivial conditions are the most practically relevant; the first-order approximate TTR condition is \emph{always} satisfied for the maximally mixed prior state. The second-order approximate TTR conditions are more restrictive so harder to satisfy. That said, it reduces to a simpler condition when the prior state is again the maximally mixed state. These results are thus expected to naturally apply to QEC scenarios in which the maximally mixed state is the relevant state, when an experimentalist wishes to recover all states in a code space without placing weight on specific states. Third, we proposed an approximate implementation of Petz recovery maps through collisional Lindbladian tabletop reverse dynamics, and demonstrated it with an excitation-exchange Hamiltonian as an application. Since collision models are a natural framework for describing thermalization and repeated system-bath interactions, we expect the collisional Lindbladian TTR condition to prove useful in quantum thermodynamics as well.

Lastly, we highlight several open questions and possible future research directions, both general and technical.
\begin{enumerate}
    \item Our definition of TTR, introduced in Ref.~\cite{aw2024role}, assumes that what is difficult is to create a new dynamics, and thus requests that the Petz map be implemented with a dilation involving $U^\dagger$. The preparation of the ancilla $\xi'$ is considered as a free operation. In the lab, whether an operation is considered free or resourceful depends strongly on the experimental platform that is being used. Thus one may want to consider alternative definitions of TTR, perhaps platform-specific ones.
    \item In all of the examples we considered, a dynamics appears to be tabletop time-reversible whenever a prior state is a steady state. We do not have a proof that this is true in general, and leave it for a future investigation.
    \item The whole inspiration for the TTR condition came from the observation that thermal operations satisfy it \cite{alhambra2018work}. In the language introduced here, they are $\infty^\texttt{th}$-order approximate TTR: the Petz recovery map is just the backward process with the same ancilla as the forward process, and regardless of the evolution time. We conjecture that thermal operations are the only operations with this property.
\end{enumerate}

\section*{Acknowledgements}
We thank Paolo Abiuso, Clive Aw, Xueyuan Hu, Stefan Nimmrichter, Matteo Scandi, Paul Skrzypczyk, Jacopo Surace and Lin Htoo Zaw for insightful discussions.

\section*{Data availability statement}
No new data were created or analysed in this study.

\section*{Funding}
M.S. and V.S. are supported by the National Research Foundation, Singapore through the National Quantum Office, hosted in A*STAR, under its Centre for Quantum Technologies Funding Initiative (S24Q2d0009); and by the Ministry of Education, Singapore, under the Tier 2 grant ``Bayesian approach to irreversibility'' (Grant No.~MOE-T2EP50123-0002). H.K. is supported by the KIAS Individual Grant No. CG085302 at Korea Institute for Advanced Study and National Research Foundation of Korea (Grants No. RS-2023-NR119931, No. RS-2024-00413957 and No. RS-2024-00438415) funded by the Korean Government (MSIT).

\section*{Author contributions}
V.S. initiated the project, M.S., H.K., and V.S. performed the research, and M.S. wrote the manuscript, edited by H.K. and V.S.

\section*{Competing interests}
The authors declare no competing interests.

\appendix
\setcounter{exm}{0}
\setcounter{lem}{0}
\setcounter{rst}{0}

\section*{Appendix}

\counterwithout{equation}{section}

\renewcommand{\theequation}{A\arabic{equation}}

\setcounter{equation}{0}

\section{Exact TTR dynamics} \label{app:example}

Here we review \cref{ex:XX} in more detail. For the sake of this appendix to be self-contained, we may repeat some content that already appeared in the main text.

Let us recall the example here: consider a dynamics defined by $H_{\texttt{tot}}=gXX$ and an ancilla $\xi\stackrel{s.d}{=}\sum_j p_k \op*{e_k}$, and choose a prior state $\gamma$ satisfying $\comm{\gamma}{X}= 0$ such that $\gamma$ is a steady state. 
Recall that 
\begin{eqnarray}
\begin{aligned}
    e^{-iXX \theta}(\gamma \ten \xi) e^{iXX \theta} = \cos^2{\theta} (\gamma \ten \xi)  + \sin^2{\theta} (\gamma \ten X \xi X)+ i(\cos{\theta}\sin{\theta}) X \gamma\ten (X\xi-\xi X),
\end{aligned}
\end{eqnarray} 
where $\theta = g \Delta t$.

First, if $\comm{\xi}{X}= 0$, then the dynamics of interest becomes a thermal operation which is already known to be TTR when $\gamma$ is a steady state. Second, if $\comm{\xi}{Z}= 0$, $\xi$ satisfies  $\Tr(\xi X)=0$. Then it becomes a probabilistic mixture of the identity channel and a unitary $X$:  
\begin{align}
    \mathcal{N}_{\theta}(\rho) :=&\Tr_{\texttt{E}} \left(e^{-iXX \theta}(\rho \ten \xi) e^{iXX \theta} \right) \\
    =& \Tr_{\texttt{E}}(\cos{\theta}\one-i\sin{\theta}XX)(\rho \ten \xi)(\cos{\theta}\one+i\sin{\theta}XX) \quad\\
    =& (\cos^2{\theta}) \rho + (\sin^2{\theta}) X\rho X + i(\cos{\theta}\sin{\theta})\Tr(X\xi)\comm{\rho}{X}\\
    =& (\cos^2{\theta}) \rho + (\sin^2{\theta}) X\rho X.
    \end{align}

Now, let us consider a non-trivial case when $\comm{\xi}{X}\ne 0$ and $\comm{\xi}{Z}\ne 0$, that is, $U_\theta=e^{-iXX \theta}$ is neither thermal operation nor mixture of unitary operations. $U_\theta$ is product preserving with respect to $\gamma,\xi$ whenever $\theta=\frac{N}{2}\pi$ for any integer $N$. After acting $U_\theta$ on $\gamma \ten \xi$, we get 
\begin{align}
    e^{-iXX \theta}(\gamma \ten \xi) e^{iXX \theta}  &= (\cos{\theta}\one-i\sin{\theta}XX)(\gamma \ten \xi)(\cos{\theta}\one+i\sin{\theta}XX)\\
    &= \gamma \ten (\cos^2{\theta} \xi  + \sin^2{\theta} X \xi X) + i(\cos{\theta}\sin{\theta}) X \gamma\ten (X\xi-\xi X),
\end{align}    
because $\gamma$ satisfies $\comm{\gamma}{X}=0$. As such, the output state is a product state, when $\theta=\frac{N}{2}\pi$. It appears that the system and the environment become correlated otherwise. Nevertheless, we can show $(H_{\texttt{tot}},\xi)$ can be TTR even when $\theta\ne\frac{N}{2}\pi$. To verify this, we compute the transition matrix  $\Phi_{mj\leftarrow nk}$:
\begin{align}
    \Phi_{mj\leftarrow nk} &\equiv 
    \bra{\lambda_m}_{\texttt{S}}\bra*{e'_j}_{\texttt{E}} U_{\theta} \ket*{\lambda_n}_{\texttt{S}}\ket*{e_k}_{\texttt{E}} \\
    &= \bra{\lambda_m}_{\texttt{S}}\bra*{e'_j}_{\texttt{E}}(\cos{\theta}\one-i\sin{\theta}XX) \ket*{\lambda_n}_{\texttt{S}}\ket*{e_k}_{\texttt{E}}\\
    &= \bra{\lambda_m}\ket*{\lambda_n}\bra*{e'_j}\big(\underbrace{\cos\theta\ket*{e_k}+i(-1)^{n+1}\sin\theta X\ket*{e_{k}}}_{=:\ket*{f_{\theta,n,k}}}\big).
\end{align}

Let us choose $\xi'$ such that its eigenbasis is given by $\ket*{e'_{j}}:=\cos\theta\ket*{e_j}- i\sin\theta X\ket*{e_{j}}=e^{-i\theta X}\ket*{e_{j}}$ and $\Tr[\xi' X]=\Tr[\xi X]$. Note that $\{\ket*{e'_j}\}$ forms a valid set of an orthonormal basis. Then we compute $\bra*{e_j'}\ket*{f_{\theta,n,k}}$ as follows,
\begin{align}
    \bra*{e_j'}\ket*{f_{\theta,n,k}} 
    &= 
    \begin{cases}
        \bra*{e_j}e^{i\theta X}\, e^{-i\theta X} \ket{e_k}, \quad &\text{if} \,\, n=0\\
        \bra*{e_j}e^{i\theta X}\, e^{i\theta X} \ket{e_k}, \quad &\text{if} \,\,n=1.
    \end{cases} \\
    &= 
    \begin{cases}
        \delta_{j,k}, \quad &\text{if} \,\, n=0\\
        \cos{2\theta}\delta_{j,k}+i\sin{2\theta}\bra*{e_j}X\ket*{e_k}, \quad &\text{if} \,\,n=1.
    \end{cases} \nonumber
\end{align}

This leads to the transition matrix given by 
\begin{align}
    \Phi_{mj\leftarrow nk} &= \bra{\lambda_m}\ket*{\lambda_n}\bra*{e'_j}\ket*{f_{\theta,n,k}}\\
    &= 
    \begin{cases}
        \delta_{m,n}\delta_{j,k}, \quad &\text{if} \,\, n=0\\
        \delta_{m,n}\left(\cos{2\theta}\delta_{j,k}+i\sin{2\theta}\bra*{e_j}X\ket*{e_k}\right), \quad &\text{if} \,\,n=1.
    \end{cases} \nonumber
\end{align}

Finally we verify the TTR condition in each of the four cases $(n_1,n_2)\in\{0,1\}^2$.\\

\noindent (a) when $n_1=n_2=0$, the TTR condition is given by 
\begin{align}
        \sum_{jk} \left(p'_j\sqrt{r'_{m_1}r'_{m_2}}-p_k\sqrt{r_{n_1}r_{n_2}}\right)\Phi_{m_1 j\leftarrow n_1k}\Phi^*_{m_2 j\leftarrow n_2k} &= r_{n_1} \sum_{jk} \left(p'_j-p_k\right)\delta_{j,k}\\
        &= r_{n_1}  \Big(\sum_j p'_j-\sum_k p_k\Big)=0.
\end{align}
(b) when $n_1=n_2=1$, we have 
\begin{align}
    &\sum_{jk} \left(p'_j\sqrt{r'_{m_1}r'_{m_2}}-p_k\sqrt{r_{n_1}r_{n_2}}\right)\Phi_{m_1 j\leftarrow n_1k}\Phi^*_{m_2 j\leftarrow n_2k} \nonumber \\
    &= r_{n_1} \sum_{jk} \left(p'_j-p_k\right)\left(\cos^2{2\theta}\delta_{j,k}+\sin^2{2\theta}\bra*{e_j}X\ket*{e_k}\!\!\bra*{e_k}X\ket*{e_j}\right)\\
    &= r_{n_1} \cos^2{2\theta} \sum_{jk} \left(p'_j-p_k\right)\delta_{j,k} + r_{n_1} \sin^2{2\theta} \sum_{jk} \left(p'_j-p_k\right)\bra*{e_j}X\ket*{e_k}\bra*{e_k}X\ket*{e_j}\\
    &= r_{n_1} \cos^2{2\theta} \left(\left(\sum_j p'_j\right)-\left(\sum_k p_k\right)\right) + r_{n_1} \sin^2{2\theta} \left(\left(\sum_j p'_j\right)-\left(\sum_k p_k\right)\right)\\
    &= 0
\end{align}
(c) when $n_1=1,n_2=0$, we have 
\begin{align}
        &\sum_{jk} \left(p'_j\sqrt{r'_{m_1}r'_{m_2}}-p_k\sqrt{r_{n_1}r_{n_2}}\right)\Phi_{m_1 j\leftarrow n_1k}\Phi^*_{m_2 j\leftarrow n_2k} \nonumber\\
        &= \sqrt{r_{n_1} r_{n_2}} \sum_{jk} \left(p'_j-p_k\right) \delta_{j,k} \left(\cos{2\theta}\delta_{j,k}+i\sin{2\theta}\bra*{e_j}X\ket*{e_k}\right)\\
        &= \cos{2\theta}\sqrt{r_{n_1} r_{n_2}} \sum_{jk} \left(p'_j-p_k\right) \delta_{j,k} + i\sin{2\theta}\sqrt{r_{n_1} r_{n_2}} \sum_{jk} \left(p'_j-p_k\right) \delta_{j,k} \bra*{e_j}X\ket*{e_k}\\
        &= \cos{2\theta}\sqrt{r_{n_1} r_{n_2}} \left(\left(\sum_j p'_j\right)-\left(\sum_k p_k\right)\right) + i\sin{2\theta}\sqrt{r_{n_1} r_{n_2}}  \left(\left(\sum_jp'_j\bra*{e_j}X\ket*{e_j}\right)-\left(\sum_{k}p_k \bra*{e_k}X\ket*{e_k}\right)\right) \\
        &= i\sin{2\theta}\sqrt{r_{n_1} r_{n_2}} \left(\sum_j p'_j
        \Tr\left(\op*{e_j} X \right)
        \right) - i\sin{2\theta}\sqrt{r_{n_1} r_{n_2}} \left(\sum_{k}p_k \Tr\left(\op*{e_k} X \right)\right) \\
        &= i\sin{2\theta}\sqrt{r_{n_1} r_{n_2}}  \Big(\Tr[\xi'X]-\Tr[\xi X]\Big)\\
        &=0,
\end{align}
where we used $\Tr(\op*{e_j} X )=\Tr(\op*{e'_j} X )$ because $\ket*{e'_j}=e^{-i\theta X}\ket*{e_j}$.\\

\noindent (d) when $n_1=0,n_1=1$, we arrive at the same conclusion similar to the case (c).

\section{First-order approximate TTR} \label{app:first_order}

Let $\rho$ be an arbitrary state. Given $(H_{\texttt{tot}},\xi)$, the Kraus representation of $\N_{\Delta t}$ is given by 
    \begin{align}
        \N_{\Delta t}(\rho) &= \sum_{j,k} p_k \bra*{e_j}_{\texttt{E}} e^{-i H_{\texttt{tot}}\Delta t}\ket*{e_k}_{\texttt{E}}  \rho \bra*{e_k}_{\texttt{E}} e^{i H_{\texttt{tot}}\Delta t}\ket*{e_j}_{\texttt{E}} \label{eq:app_Kraus} \\
        &=  \sum_{j,k} p_k  \bra*{e_j}_{\texttt{E}} \left( \one -i\Delta t H_{\texttt{tot}} \right)\ket*{e_k}_{\texttt{E}} \rho \bra*{e_k}_{\texttt{E}} \left( \one +i\Delta t H_{\texttt{tot}} \right)\ket*{e_j}_{\texttt{E}} + \order{(\Delta t)^2} \\
        &= \rho - i\Delta{t}\sum_{k} p_k  \bra*{e_k}_{\texttt{E}} H_{\texttt{tot}} \ket*{e_k}_{\texttt{E}} \rho +i\Delta{t}\rho\sum_{k} p_k  \bra*{e_k}_{\texttt{E}} H_{\texttt{tot}} \ket*{e_k}_{\texttt{E}}  + \order{(\Delta t)^2}\\
        &= \rho - i\Delta{t}\Tr_{\texttt{E}} \left( H_{\texttt{tot}} ( \one_{\texttt{S}} \otimes \xi_{\texttt{E}}) \right) \rho +i\Delta{t}\rho\Tr_{\texttt{E}} \left( H_{\texttt{tot}} ( \one_{\texttt{S}} \otimes \xi_{\texttt{E}}) \right)  + \order{(\Delta t)^2}\label{eq:app_first_forward}
    \end{align}

\subsection{Tabletop reverse maps}
Similarly, we expand $\hat{\N}_{\Delta t,\xi'}^{(\texttt{T})}$ up to the first-order of $\Delta t$:
    \begin{align}
        \hat{\N}_{\Delta t,\xi'}^{(\texttt{T})}(\rho) &= \sum_{j,k} p'_j  \bra*{e_k'}_{\texttt{E}} e^{i H_{\texttt{tot}}\Delta t}\ket*{e_j'}_{\texttt{E}} \rho \bra*{e_j'}_{\texttt{E}} e^{-i H_{\texttt{tot}}\Delta t}\ket*{e_k'}_{\texttt{E}}\\
        &=  \sum_{j,k} p'_j  \bra*{e_k'}_{\texttt{E}} \left( \one +i\Delta t H_{\texttt{tot}} \right)\ket*{e_j'}_{\texttt{E}} \rho \bra*{e_j'}_{\texttt{E}} \left( \one -i\Delta t H_{\texttt{tot}} \right)\ket*{e_k'}_{\texttt{E}} + \order{(\Delta t)^2} \\
        &= \rho + i\Delta{t}\sum_{j} p'_j  \bra*{e_j'}_{\texttt{E}}  H_{\texttt{tot}} \ket*{e_j'}_{\texttt{E}} \rho -i\Delta{t}\rho\sum_{j} p'_j  \bra*{e_j'}_{\texttt{E}} H_{\texttt{tot}} \ket*{e_j'}_{\texttt{E}}  + \order{(\Delta t)^2}\\
        &=\rho+i\Delta t\Tr_{\texttt{E}} \left( H_{\texttt{tot}} ( \one \otimes \xi') \right)\rho-i\Delta t\rho\Tr_{\texttt{E}} \left( H_{\texttt{tot}} (\one \otimes \xi') \right)+ \order{(\Delta t)^2}.
    \end{align}

\subsection{Petz recovery maps}
In turn, the Kraus representation of the Petz recovery map is given by 
    \begin{align}
        \hat{{\cal N}}^{(\texttt{P})}_{\Delta t,\gamma}(\rho)=\sum_{j,k} p_k \gamma^{\frac{1}{2}} \!\bra*{e_k}_{\texttt{E}} e^{i H_{\texttt{tot}}\Delta t}\ket*{e_j}_{\texttt{E}} \!\gamma_{\Delta t}^{-\frac{1}{2}} \rho \gamma_{\Delta t}^{-\frac{1}{2}} \!\bra*{e_j}_{\texttt{E}} e^{-i H_{\texttt{tot}}\Delta t}\ket*{e_k}_{\texttt{E}} \! \gamma^{\frac{1}{2}}, 
    \end{align}
    where we denote $\gamma_{\Delta t}\equiv\N_{\Delta t}(\gamma)$.  Inserting $\rho:=\gamma$ into Eq.~\eqref{eq:app_first_forward}, we can represent $\gamma_{\Delta t}$ as 
    \begin{align}
        \gamma_{\Delta t}=\gamma-i\Delta t\comm{H}{\gamma}+\order{(\Delta t)^2}, \label{eq:app_gammat}
    \end{align}
    where $H\equiv\Tr_{\texttt{E}} \left( H_{\texttt{tot}} ( \one_{\texttt{S}} \otimes \xi_{\texttt{E}}) \right)$.
    
    Define $W$ as $W:=\gamma_{\Delta t}^{-\frac{1}{2}}$, and let $A$ be such that $\gamma^{-\frac{1}{2}}A\gamma^{-\frac{1}{2}}$ is the first-order term in the power series of $W$ in terms of $\Delta t$, i.e.,
    \begin{align}
        W&:=\gamma_{\Delta t}^{-\frac{1}{2}}\\
        &=\left(\gamma-i\Delta t\comm{H}{\gamma}+\order{(\Delta t)^2}\right)^{-\frac{1}{2}}\\
        &=\gamma^{-\frac{1}{2}} + \Delta t\gamma^{-\frac{1}{2}}A\gamma^{-\frac{1}{2}}+\order{(\Delta t)^2}.
    \end{align}
    
    We then expand $W^{-1}$ in $\Delta t$ and get 
    \begin{align}
    W^{-1} =& \left(\gamma^{-\frac{1}{2}} +\Delta t \gamma^{-\frac{1}{2}} A \gamma^{-\frac{1}{2}} \right)^{-1} + \order{(\Delta t)^2} \\
    =& \left(\one + \Delta t  A \gamma^{-\frac{1}{2}} \right)^{-1}\gamma^{\frac{1}{2}} + \order{(\Delta t)^2} \\
    =& \left(\one - \Delta t  A \gamma^{-\frac{1}{2}} \right)\gamma^{\frac{1}{2}}  + \order{(\Delta t)^2}\\
    =& \gamma^{\frac{1}{2}} -\Delta t A   + \order{(\Delta t)^2}.
\end{align}

Using this, we evaluate $W^{-2}$:
\begin{align}
    W^{-2}=& \left(W^{-1}\right)^2\\
    =& \left(\gamma^{\frac{1}{2}} -\Delta t A\right)^2 + \order{(\Delta t)^2}\\
    =& \gamma - \Delta t \left( \gamma^{\frac{1}{2}}A + A \gamma^{\frac{1}{2}} \right) + \order{(\Delta t)^2}. \label{eq:app_w}
\end{align}

Since $W^{-2}=\gamma_{\Delta t}$, the first-order terms of Eq.~\eqref{eq:app_gammat} and Eq.~\eqref{eq:app_w} must match. From this, we get 
\begin{align}
    \gamma^{\frac{1}{2}}A + A \gamma^{\frac{1}{2}} = -i\comm{H}{\gamma}.
\end{align}

Observe $A := -i\int_0^\infty e^{-x \sqrt{\gamma}} \comm{H}{\gamma} e^{-x \sqrt{\gamma}} dx $ is a solution of the above equation:
\begin{align}
    \sqrt{\gamma}A + A\sqrt{\gamma}&= -i\int_0^\infty \left(\sqrt{\gamma} e^{-x \sqrt{\gamma}} \comm{H}{\gamma} e^{-x \sqrt{\gamma}}+ e^{-x \sqrt{\gamma}} \comm{H}{\gamma} e^{-x \sqrt{\gamma}} \sqrt{\gamma} \right)dx \\
    &= i\int_0^\infty \left(\frac{d}{dx}e^{-x \sqrt{\gamma}} \comm{H}{\gamma} e^{-x \sqrt{\gamma}}\right)dx\\
    &= i e^{-x \sqrt{\gamma}} \comm{H}{\gamma} e^{-x \sqrt{\gamma}} \eval_{x=0}^{x=\infty}\\
    &=-i\comm{H}{\gamma}.
\end{align}

Combining altogether, we can express $\hat{{\cal N}}^{(\texttt{P})}_{\Delta t,\gamma}(\rho)$ up to first order of $\Delta t$ as
    \begin{align}
    \begin{aligned}
        \hat{{\cal N}}^{(\texttt{P})}_{\Delta t,\gamma}(\rho)&=\rho+i\Delta t\Tr_{\texttt{E}} \left( (\gamma_{\texttt{S}}^{\frac{1}{2}}\ten \one_{\texttt{E}} ) H_{\texttt{tot}} ( \gamma_{\texttt{S}}^{-\frac{1}{2}} \ten \xi_{\texttt{E}}) \right)\rho +\Delta tA\gamma^{-\frac{1}{2}}\rho\\
        &\quad -i\Delta t\rho\Tr_{\texttt{E}} \left( (\gamma_{\texttt{S}}^{-\frac{1}{2}}\ten \one_{\texttt{E}} ) H_{\texttt{tot}} ( \gamma_{\texttt{S}}^{\frac{1}{2}} \ten \xi_{\texttt{E}})\right)+\Delta t\rho\gamma^{-\frac{1}{2}}A+\order{(\Delta t)^2}. 
    \end{aligned}
    \end{align}

    Thus, we conclude that the first order terms of $\hat{{\cal N}}^{(\texttt{P})}_{\Delta t,\gamma}$ and $\hat{\N}_{\Delta t,\xi'_t}^{(\texttt{T})}$ match, if it holds that
    \begin{align*}
    \begin{aligned}
        \!\!&\Tr_{\texttt{E}} \left( (\gamma_{\texttt{S}}^{\frac{1}{2}}\ten \one_{\texttt{E}} ) H_{\texttt{tot}} ( \gamma_{\texttt{S}}^{-\frac{1}{2}} \ten \xi_{\texttt{E}}) \right) - iA\gamma^{-\frac{1}{2}} = \Tr_{\texttt{E}} \left( H_{\texttt{tot}} ( \one_{\texttt{S}} \otimes \xi'_{\texttt{E}}) \right)\\
        \!\!&\Tr_{\texttt{E}} \left( (\gamma_{\texttt{S}}^{-\frac{1}{2}}\ten \one_{\texttt{E}} ) H_{\texttt{tot}} ( \gamma_{\texttt{S}}^{\frac{1}{2}} \ten \xi_{\texttt{E}}) \right) + i\gamma^{-\frac{1}{2}}A = \Tr_{\texttt{E}} \left( H_{\texttt{tot}} ( \one_{\texttt{S}} \otimes \xi'_{\texttt{E}}) \right).
    \end{aligned}
    \end{align*}

If $\gamma$ is a steady state or the maximally mixed state, $A$ vanishes.

\section{Second-order approximate TTR} \label{app:second_order}

We proceed to prove \cref{rst:secondTTR} using the similar technique used in the previous section but here we expand two maps up to the second-order of the evolution time. Let us begin with the second-order approximation of $\N_{\Delta t}$:

    \begin{align}
    \N_{\Delta t}(\rho) =&\Tr_{\texttt{E}}\left(e^{-i H_{\texttt{tot}}{\Delta t} }(\rho_{\texttt{S}} \ten \xi_{\texttt{E}})e^{i H_{\texttt{tot}}{\Delta t}} \right) \\
    =&  \Tr_{\texttt{E}}\left((\one-iH_{\texttt{tot}}{\Delta t}-H_{\texttt{tot}}^2\frac{(\Delta t)^2}{2}+\ldots)(\rho_{\texttt{S}} \ten \xi_{\texttt{E}})(\one+iH_{\texttt{tot}}{\Delta t}-H_{\texttt{tot}}^2\frac{(\Delta t)^2}{2}+\ldots)\right)\\
    =& \Tr_{\texttt{E}}\left(-i{\Delta t}\comm{H_{\texttt{tot}}}{\rho_{\texttt{S}}\ten\xi_{\texttt{E}}}+(\Delta t)^2 H_{\texttt{tot}}(\rho_{\texttt{S}}\ten\xi_{\texttt{E}})H_{\texttt{tot}}- \frac{(\Delta t)^2}{2}\acomm{H_{\texttt{tot}}^2}{\rho_{\texttt{S}}\ten\xi_{\texttt{E}}}\right) + \order{(\Delta t)^3}\\
    =& -i{\Delta t}\Tr_{\texttt{E}}\comm{H_{\texttt{tot}}}{\rho_{\texttt{S}}\ten\xi_{\texttt{E}}}+(\Delta t)^2 \Tr_{\texttt{E}}\left(H_{\texttt{tot}}(\rho_{\texttt{S}}\ten\xi_{\texttt{E}})H_{\texttt{tot}}\right)- \frac{(\Delta t)^2}{2}\Tr_{\texttt{E}}\acomm{H_{\texttt{tot}}^2}{\rho_{\texttt{S}}\ten\xi_{\texttt{E}}}+ \order{(\Delta t)^3}\\
    =& -i{\Delta t}\comm{\Tr_{\texttt{E}}(H_{\texttt{tot}}\xi_{\texttt{E}})}{\rho_{\texttt{S}}}+(\Delta t)^2 \sum_{jk} p_k \bra*{f_j}_{\texttt{E}} H_{\texttt{tot}} \ket*{e_k}_{\texttt{E}} \rho_{\texttt{S}} \bra*{e_k}_{\texttt{E}} H_{\texttt{tot}} \ket*{f_j}_{\texttt{E}}\\
    \quad &
    - \frac{(\Delta t)^2}{2}\sum_{jk}p_k \acomm{ \bra*{e_k}_{\texttt{E}} H_{\texttt{tot}} \ket*{f_j}_{\texttt{E}} \bra*{f_j}_{\texttt{E}} H_{\texttt{tot}} \ket*{e_k}_{\texttt{E}}}{\rho_{\texttt{S}}}+ \order{(\Delta t)^3}\nonumber\\
    =& \Delta t \, \mathcal{H}{\rho} + (\Delta t)^2 \sum_{jk}p_k\mathcal{D}[L_{jk}] \rho + \order{(\Delta t)^3},
\end{align}
defining the superoperators 
\begin{align}
    \mathcal{H}\rho \equiv& -i\comm{H}{\rho}, \qand \\ 
    \mathcal{D}[L_{jk}]\rho\equiv& \left( L_{jk} \rho L_{jk}^\dagger - \frac{1}{2}\acomm{L_{jk}^\dagger L_{jk}}{\rho}\right), \label{eq:app_D}
\end{align}
where $H\equiv\Tr_{\texttt{E}} \left( H_{\texttt{tot}} ( \one_{\texttt{S}} \otimes \xi_{\texttt{E}}) \right)$  
and $L_{jk}\equiv  (\one_{\texttt{S}}\ten\bra*{f_j}_{\texttt{E}}){H_{\texttt{tot}}}(\one_{\texttt{S}}\ten\ket{e_k}_{\texttt{E}})$ for any orthonormal basis $\{\ket*{f_j}\}_j$. For our purpose, we choose $\ket*{f_j}=\ket*{e'_j}$ the eigenbasis of $\xi'$. We use $\mathcal{H}$ and $\mathcal{D}$ for notational brevity, but this should not be confused with the components of the Lindbladian of $\N_{\Delta t}$. Recall that we use $\comm{X}{Y}\equiv XY-YX$ to represent the commutator of $X$ and $Y$, and $\acomm{X}{Y}\equiv XY+YX$ to represent the anti-commutator of $X$ and $Y$.

\subsection{Tabletop reverse maps}

Similarly, we get the second-order approximation of the tabletop reverse map $\hat{{\cal N}}^{(\texttt{T})}_{\Delta t,\xi'}$:
\begin{align}
    \hat{{\cal N}}^{(\texttt{T})}_{\Delta t,\xi'}(\rho) = \rho +\Delta t\hat{\mathcal{H}}^{(\texttt{T})}_{\xi'}\rho + (\Delta t)^2 \sum_{jk} p'_j \hat{\mathcal{D}}^{(\texttt{T})}_{\xi'}[L_{jk}^\dagger]\rho + \order{(\Delta t)^3},
\end{align}
defining the superoperators
\begin{align}
    \hat{\mathcal{H}}^{(\texttt{T})}_{\xi'}\rho \equiv& - i \comm{H^{(\texttt{T})}_{\xi'}}{\rho}, \qand \\ 
    \hat{\mathcal{D}}^{(\texttt{T})}_{\xi'}[L_{jk}^\dagger]\rho\equiv& \left( L_{jk}^\dagger \rho L_{jk} - \frac{1}{2}\acomm{L_{jk} L_{jk}^\dagger}{\rho}\right), 
\end{align}
where $H^{(\texttt{T})}_{\xi'}\equiv -\Tr_{\texttt{E}} \left( H_{\texttt{tot}} ( \one_{\texttt{S}} \otimes \xi'_{\texttt{E}}) \right)$. 

\subsection{Petz recovery maps}

For the second-order approximation of $\hat{{\cal N}}^{(\texttt{P})}_{\Delta t,\gamma}$, we need two ingredients; ${\cal N}_{\Delta t}^\dagger$ and $W:=\left(\N_{\Delta t}(\gamma)\right)^{-\frac{1}{2}}$.

From the Kraus representation of $\N_{\Delta t}$ in Eq.~\eqref{eq:app_Kraus}, we get 
\begin{align}
    {\cal N}_{\Delta t}^\dagger(\rho) &= \sum_{j,k} p_k  \bra*{e_k}_{\texttt{E}} e^{i H_{\texttt{tot}}\Delta t}\ket*{e_j}_{\texttt{E}} \rho \bra*{e_j}_{\texttt{E}} e^{-i H_{\texttt{tot}}\Delta t}\ket*{e_k}_{\texttt{E}}\\
    &=\rho -{\Delta t} \mathcal{H}\rho
    + (\Delta t)^2 \sum_{jk} p_k \mathcal{D}[L_{jk}^\dagger] \rho + \order{(\Delta t)^3}.
\end{align}

Now, we need to evaluate $W$ up to the second order of $\Delta t$:
\begin{align}
    W&:=\frac{1}{\sqrt{\N_{\Delta t}(\gamma)}}\\
    &=\left(\gamma+ \Delta t \mathcal{H}\gamma + (\Delta t)^2 \sum_{jk}p_k\mathcal{D}[L_{jk}]\gamma + \order{(\Delta t)^3}\right)^{-\frac{1}{2}}
\end{align}

Let $A, B$ such that
\begin{align}
    W &= \gamma^{-\frac{1}{2}} +(\Delta t) \gamma^{-\frac{1}{2}} A \gamma^{-\frac{1}{2}} + (\Delta t)^2\gamma^{-\frac{1}{2}} B \gamma^{-\frac{1}{2}} + \order{(\Delta t)^3} \\
    &= \gamma^{-\frac{1}{2}}\left(\one +(\Delta t)  A \gamma^{-\frac{1}{2}} + (\Delta t)^2 B \gamma^{-\frac{1}{2}}\right) + \order{(\Delta t)^3}
\end{align}

We first evaluate $W^{-1}$:
\begin{align}
    W^{-1} &= \left(\one +(\Delta t)  A \gamma^{-\frac{1}{2}} + (\Delta t)^2 B \gamma^{-\frac{1}{2}}\right)^{-1}\gamma^{\frac{1}{2}} + \order{(\Delta t)^3} \\
    &= \left(\one -(\Delta t)  A \gamma^{-\frac{1}{2}} - (\Delta t)^2 B \gamma^{-\frac{1}{2}} + (\Delta t)^2 A\gamma^{-\frac{1}{2}}A\gamma^{-\frac{1}{2}}\right)\gamma^{\frac{1}{2}}  + \order{(\Delta t)^3}\\
    &= \gamma^{\frac{1}{2}} -(\Delta t) A  + (\Delta t)^2 \left(A\gamma^{-\frac{1}{2}}A -B\right) + \order{(\Delta t)^3}.
\end{align}

Then, we obtain 
\begin{align}
    W^{-2} =& \left(W^{-1}\right)^2\\
    =& \left(\gamma^{\frac{1}{2}} -(\Delta t) A  + (\Delta t)^2 \left(A\gamma^{-\frac{1}{2}}A -B\right)\right)^2 + \order{(\Delta t)^3}\\
    =& \gamma -(\Delta t)\left( \gamma^{\frac{1}{2}}A + A\gamma^{\frac{1}{2}} \right) +(\Delta t)^2 \left( A^2 + \gamma^{\frac{1}{2}}A\gamma^{-\frac{1}{2}}A + A\gamma^{-\frac{1}{2}}A\gamma^{\frac{1}{2}} -\gamma^{\frac{1}{2}}B - B\gamma^{\frac{1}{2}} \right).
\end{align}

In \cref{app:first_order}, we have found 
\begin{align}
    A = -i\int_0^\infty e^{-x \sqrt{\gamma}} \comm{H}{\gamma} e^{-x \sqrt{\gamma}} dx.
\end{align}
Similarly, from the fact that the second-order terms of $W^{-2}$ and $\N_{\Delta t}(\gamma)$ must match since $W^{-2}=\N_{\Delta t}(\gamma)$ by definition, we have a Lyapunov equation (special case of a Sylvester equation) 

\begin{align}
\begin{aligned}
    \gamma^{\frac{1}{2}}B + B\gamma^{\frac{1}{2}} &= A^2 + \gamma^{\frac{1}{2}}A\gamma^{-\frac{1}{2}}A + A\gamma^{-\frac{1}{2}}A\gamma^{\frac{1}{2}}-\sum_{jk}p_k\mathcal{D}[L_{jk}] \gamma\\
    &=:C, \label{eq:app_B}
\end{aligned}
\end{align}
and get a solution
\begin{align}
    B = \int_{0}^{\infty} e^{-x\sqrt{\gamma}} C e^{-x\sqrt{\gamma}} dx.
\end{align}

Having these, we get
\begin{align}
    \frac{1}{\sqrt{\N_{\Delta t}(\gamma)}} \rho \frac{1}{\sqrt{\N_{\Delta t}(\gamma)}}&=\left(\gamma^{-\frac{1}{2}} +(\Delta t) \gamma^{-\frac{1}{2}} A \gamma^{-\frac{1}{2}} + (\Delta t)^2\gamma^{-\frac{1}{2}} B \gamma^{-\frac{1}{2}}\right)\rho\left(\gamma^{-\frac{1}{2}} +(\Delta t) \gamma^{-\frac{1}{2}} A \gamma^{-\frac{1}{2}} + (\Delta t)^2\gamma^{-\frac{1}{2}} B \gamma^{-\frac{1}{2}}\right) + \order{(\Delta t)^3}\\
    &= \gamma^{-\frac{1}{2}}\rho\gamma^{-\frac{1}{2}} + (\Delta t) \left(\gamma^{-\frac{1}{2}}\rho \gamma^{-\frac{1}{2}} A \gamma^{-\frac{1}{2}} + \gamma^{-\frac{1}{2}} A \gamma^{-\frac{1}{2}} \rho\gamma^{-\frac{1}{2}} \right) + (\Delta t)^2 \gamma^{-\frac{1}{2}} A \gamma^{-\frac{1}{2}} \rho \gamma^{-\frac{1}{2}} A \gamma^{-\frac{1}{2}} \nonumber\\
    &\qquad + (\Delta t)^2 \left(\gamma^{-\frac{1}{2}} \rho \gamma^{-\frac{1}{2}} B \gamma^{-\frac{1}{2}} + \gamma^{-\frac{1}{2}} B \gamma^{-\frac{1}{2}}\rho\gamma^{-\frac{1}{2}}\right) + \order{(\Delta t)^3}.
\end{align}

Then, it follows that 
\begin{align}
    \N_{\Delta t}^\dagger\left( \frac{1}{\sqrt{\N_{\Delta t}(\gamma)}} \rho \frac{1}{\sqrt{\N_{\Delta t}(\gamma)}}\right) &=\gamma^{-\frac{1}{2}}\rho\gamma^{-\frac{1}{2}} - {\Delta t} \mathcal{H}\gamma^{-\frac{1}{2}}\rho\gamma^{-\frac{1}{2}} + (\Delta t)^2 \sum_{jk}p_k\mathcal{D}[L_{jk}^\dagger] \gamma^{-\frac{1}{2}}\rho\gamma^{-\frac{1}{2}} \nonumber \\
    &\quad + (\Delta t)\left(\gamma^{-\frac{1}{2}}\rho \gamma^{-\frac{1}{2}} A \gamma^{-\frac{1}{2}} + \gamma^{-\frac{1}{2}} A \gamma^{-\frac{1}{2}} \rho\gamma^{-\frac{1}{2}}\right) - (\Delta t)^2 \mathcal{H}\gamma^{-\frac{1}{2}}\rho \gamma^{-\frac{1}{2}} A \gamma^{-\frac{1}{2}} - (\Delta t)^2 \mathcal{H}\gamma^{-\frac{1}{2}} A \gamma^{-\frac{1}{2}} \rho\gamma^{-\frac{1}{2}} \nonumber \\
    &\quad + (\Delta t)^2 \gamma^{-\frac{1}{2}} A \gamma^{-\frac{1}{2}} \rho \gamma^{-\frac{1}{2}} A \gamma^{-\frac{1}{2}} + (\Delta t)^2 \left(\gamma^{-\frac{1}{2}} \rho \gamma^{-\frac{1}{2}} B \gamma^{-\frac{1}{2}} + \gamma^{-\frac{1}{2}} B \gamma^{-\frac{1}{2}}\rho\gamma^{-\frac{1}{2}}\right) + \order{(\Delta t)^3}. 
\end{align}

Combining altogether, we obtain
\begin{align}
    \hat{{\cal N}}^{(\texttt{P})}_{\Delta t,\gamma}(\rho) &\equiv \sqrt{\gamma}\N_{\Delta t}^\dagger\left( \frac{1}{\sqrt{\N_{\Delta t}(\gamma)}} \rho \frac{1}{\sqrt{\N_{\Delta t}(\gamma)}}\right)\sqrt{\gamma} \\
    &=\rho - {\Delta t} \gamma^{\frac{1}{2}}\left(\mathcal{H} \gamma^{-\frac{1}{2}}\rho\gamma^{-\frac{1}{2}}\right)\gamma^{\frac{1}{2}} + (\Delta t)^2 \sum_{jk}p_k\gamma^{\frac{1}{2}}\left(\mathcal{D}[L_{jk}^\dagger] \gamma^{-\frac{1}{2}}\rho\gamma^{-\frac{1}{2}}\right)\gamma^{\frac{1}{2}} \nonumber \\
    &\quad + (\Delta t)\left(\rho \gamma^{-\frac{1}{2}} A  +  A \gamma^{-\frac{1}{2}} \rho \right) - (\Delta t)^2 \gamma^{\frac{1}{2}}\left(\mathcal{H} \gamma^{-\frac{1}{2}}\rho \gamma^{-\frac{1}{2}} A \gamma^{-\frac{1}{2}}\right)\gamma^{\frac{1}{2}} - (\Delta t)^2 \gamma^{\frac{1}{2}}\left(\mathcal{H} \gamma^{-\frac{1}{2}} A \gamma^{-\frac{1}{2}} \rho\gamma^{-\frac{1}{2}}\right)\gamma^{\frac{1}{2}}  \nonumber\\
    &\quad + (\Delta t)^2  A \gamma^{-\frac{1}{2}} \rho \gamma^{-\frac{1}{2}} A  + (\Delta t)^2 \left( \rho \gamma^{-\frac{1}{2}} B  +  B \gamma^{-\frac{1}{2}}\rho\right) + \order{(\Delta t)^3}. 
\end{align}

The second-order matching of two maps then gives the following condition 
\begin{align}
    &\sum_{jk} p_k \left( \gamma^{\frac{1}{2}}L_{jk}^\dagger\gamma^{-\frac{1}{2}} \rho \gamma^{-\frac{1}{2}}L_{jk}\gamma^{\frac{1}{2}}\right) -\frac{1}{2}\sum_{jk} p_k \left( \gamma^{\frac{1}{2}}L_{jk} L_{jk}^\dagger\gamma^{-\frac{1}{2}}\rho + \rho \gamma^{-\frac{1}{2}} L_{jk} L_{jk}^\dagger\gamma^{\frac{1}{2}}\right)
    \nonumber\\ 
    &\quad -\gamma^{\frac{1}{2}}\left(\mathcal{H} \gamma^{-\frac{1}{2}}\rho \gamma^{-\frac{1}{2}} A \gamma^{-\frac{1}{2}}\right)\gamma^{\frac{1}{2}} -\gamma^{\frac{1}{2}}\left(\mathcal{H} \gamma^{-\frac{1}{2}} A \gamma^{-\frac{1}{2}} \rho\gamma^{-\frac{1}{2}}\right)\gamma^{\frac{1}{2}} + A \gamma^{-\frac{1}{2}} \rho \gamma^{-\frac{1}{2}} A  +  \left( \rho \gamma^{-\frac{1}{2}} B  +  B \gamma^{-\frac{1}{2}}\rho\right),\nonumber\\
    &\quad \quad  =\sum_{jk} p'_j \left( L_{jk}^\dagger \rho L_{jk} - \frac{1}{2}L_{jk} L_{jk}^\dagger\rho - \frac{1}{2}\rho L_{jk} L_{jk}^\dagger\right)\label{eq:app_secondcon}
\end{align}
for any $\rho$.

If $\gamma=\one/d$ is the maximally mixed state, then $A$ vanishes and from Eq.~\eqref{eq:app_B} we get 
\begin{align}
    &\gamma^{-\frac{1}{2}}B = B\gamma^{-\frac{1}{2}} = \sqrt{d} B\\
    &\sqrt{d} B = -\frac{1}{2} \sum_{jk} p_k \left(L_{jk}L_{jk}^\dagger - L_{jk}^\dagger L_{jk}\right).
\end{align}
Together with $L_{jk}\equiv  (\one_{\texttt{S}}\ten\bra*{e'_j}_{\texttt{E}}){H_{\texttt{tot}}}(\one_{\texttt{S}}\ten\ket{e_k}_{\texttt{E}})$, the condition \eqref{eq:app_secondcon} reduces to 
\begin{align}
\begin{aligned}
    &\Tr_{\texttt{E}}\left(  H_{\texttt{tot}}(\rho_{\texttt{S}}\ten \one_{\texttt{E}} )H_{\texttt{tot}} (\one_{\texttt{S}}\ten \xi_{\texttt{E}}) \right)-\frac{1}{2}\Tr_{\texttt{E}}\left(  H_{\texttt{tot}}^2 (\one_{\texttt{S}} \ten \xi_{\texttt{E}})  \right)\rho  -\frac{1}{2}\rho\Tr_{\texttt{E}}\left(  H_{\texttt{tot}}^2 (\one_{\texttt{S}} \ten \xi_{\texttt{E}})  \right) \\
    &\quad -\frac{1}{2} \rho \left(\Tr_{\texttt{E}}\left(H_{\texttt{tot}}(\one_{\texttt{S}}\ten \xi_{\texttt{E}} )H_{\texttt{tot}}\right) \right) 
    -\frac{1}{2} \left(\Tr_{\texttt{E}}\left(H_{\texttt{tot}}^2(\one_{\texttt{S}}\ten \xi_{\texttt{E}} )\right) \right)\rho 
    \\
    &= \Tr_{\texttt{E}}\left(  H_{\texttt{tot}}(\rho_{\texttt{S}}\ten \xi'_{\texttt{E}} )H_{\texttt{tot}} \right) -\frac{1}{2}\Tr_{\texttt{E}}\left(  H_{\texttt{tot}}^2 (\one_{\texttt{S}} \ten \xi'_{\texttt{E}}) \right)\rho -\frac{1}{2}\rho\Tr_{\texttt{E}}\left(  H_{\texttt{tot}}^2 (\one_{\texttt{S}} \ten \xi'_{\texttt{E}}) \right).
\end{aligned}
\end{align}

If $\gamma$ is a steady state, then both $A$ and $B$ vanish, thus the condition \eqref{eq:app_secondcon} simplifies to 
\begin{align}
\begin{aligned}
    &\Tr_{\texttt{E}}\left( (\gamma^{\frac{1}{2}}_{\texttt{S}}\ten \one_{\texttt{E}} ) H_{\texttt{tot}}(\gamma^{-\frac{1}{2}}_{\texttt{S}}\rho_{\texttt{S}}\gamma^{-\frac{1}{2}}_{\texttt{S}}\ten \one_{\texttt{E}} )H_{\texttt{tot}} (\gamma^{\frac{1}{2}}_{\texttt{S}}\ten \xi_{\texttt{E}}) \right)\\
    &\quad  -\frac{1}{2}\Tr_{\texttt{E}}\left( \gamma^{\frac{1}{2}}_{\texttt{S}} H_{\texttt{tot}}^2 (\one_{\texttt{S}} \ten \xi_{\texttt{E}}) \gamma^{-\frac{1}{2}}_{\texttt{S}} \right)\rho\\
    &\quad  -\frac{1}{2}\rho\Tr_{\texttt{E}}\left( \gamma^{-\frac{1}{2}}_{\texttt{S}} H_{\texttt{tot}}^2 (\one_{\texttt{S}} \ten \xi_{\texttt{E}}) \gamma^{\frac{1}{2}}_{\texttt{S}} \right)\\
    &= \Tr_{\texttt{E}}\left(  H_{\texttt{tot}}(\rho_{\texttt{S}}\ten \xi'_{\texttt{E}} )H_{\texttt{tot}} \right)\\
    &\quad  -\frac{1}{2}\Tr_{\texttt{E}}\left(  H_{\texttt{tot}}^2 (\one_{\texttt{S}} \ten \xi'_{\texttt{E}}) \right)\rho\\
    &\quad  -\frac{1}{2}\rho\Tr_{\texttt{E}}\left(  H_{\texttt{tot}}^2 (\one_{\texttt{S}} \ten \xi'_{\texttt{E}}) \right).
\end{aligned}
\end{align}
If a steady state $\gamma$ satisfies $\comm{\gamma_{\texttt{S}}}{H_{\texttt{tot}}}=0$, we get the further simplified condition 
\begin{align}
\begin{aligned}
    &\Tr_{\texttt{E}}\left( H_{\texttt{tot}}(\rho_{\texttt{S}}\ten \one_{\texttt{E}} )H_{\texttt{tot}} (\one_{\texttt{S}}\ten \xi_{\texttt{E}}) \right) -\frac{1}{2}\Tr_{\texttt{E}}\left(  H_{\texttt{tot}}^2 (\one_{\texttt{S}} \ten \xi_{\texttt{E}})  \right)\rho -\frac{1}{2}\rho\Tr_{\texttt{E}}\left(  H_{\texttt{tot}}^2 (\one_{\texttt{S}} \ten \xi_{\texttt{E}})  \right)\\
    &= \Tr_{\texttt{E}}\left(  H_{\texttt{tot}}(\rho_{\texttt{S}}\ten \xi'_{\texttt{E}} )H_{\texttt{tot}} \right)  -\frac{1}{2}\Tr_{\texttt{E}}\left(  H_{\texttt{tot}}^2 (\one_{\texttt{S}} \ten \xi'_{\texttt{E}}) \right)\rho -\frac{1}{2}\rho\Tr_{\texttt{E}}\left(  H_{\texttt{tot}}^2 (\one_{\texttt{S}} \ten \xi'_{\texttt{E}}) \right).
\end{aligned}
\end{align}

\section{Lindbladian approximate TTR}
\subsection{Proof of \cref{lem:Petz_Lindblad}} \label{app:proof_lem}

Consider a quantum dynamics ${\cal N}_{\Delta t} = e^{\Delta t {\cal L}}$ for a time $\Delta t$, given by the Lindbladian 
\begin{equation}
    {\cal L}\rho = - i [H, \rho] + \sum_\mu \mathcal{D}[L_\mu]\rho.
\end{equation}
Following Ref.~\cite{kwon2022reversing}, we compute a Lindbladian $\hat{{\cal L}}^{(\texttt{P})}_{\gamma_t}$
\begin{equation}
    \hat{{\cal L}}_{\gamma_t}^{(\texttt{P})}(\bullet) = - i [H^{(\texttt{P})}_{\gamma_t}, \bullet] + \sum_\mu \mathcal{D}[L_\mu^{(\texttt{P})}(\gamma_t)](\bullet),
\end{equation}
where 
\begin{align}
    &H^{(\texttt{P})}_{\gamma_t} = -H + H_{\texttt{C}}\\
    &L_\mu^{(\texttt{P})}(\gamma_t) = \gamma_t^{1/2} L_{\mu}^\dagger \gamma_t^{-1/2},
\end{align}
with the correction Hamiltonian $H_{\texttt{C}}=H_{\texttt{C},\gamma_t}$ defined as in Eq.~\eqref{eq:hc}.

Then, the diamond distance between the Petz recovery map and the Lindbladian dynamics from $\hat{{\cal L}}_{\gamma_t}^{(\texttt{P})}$ is given as 
    \begin{equation}
        \frac{1}{2} \left \| \hat{{\cal N}}^{(\texttt{P})}_{\Delta t, \gamma_t} - e^{ \Delta t \hat{{\cal L}}^{(\texttt{P})}_{\gamma_t}} \right\|_\diamond = {\cal O}\left( (\Delta t)^2 \right).    
    \end{equation}

\begin{proof}
Let $\rho_t$ be an arbitrary state. We note that the Petz recovery map $\hat{{\cal N}}^{(\texttt{P})}_{\Delta t, \gamma_t}$ is defined as
    \begin{equation}
        \hat{{\cal N}}^{(\texttt{P})}_{\Delta t, \gamma_t}(\rho_t) \equiv \gamma_t^{\frac{1}{2}} e^{\Delta t{\cal L}^\dagger} \left( \gamma_{t+\Delta t}^{-\frac{1}{2}} \rho_t \gamma_{t+\Delta t}^{-\frac{1}{2}} \right) \gamma_t^{\frac{1}{2}},
    \end{equation}
where $\gamma_{t+\Delta t} = \N_{\Delta t}(\gamma_t)=e^{\Delta t {\cal L}} (\gamma_t)$. We express each term up to the first order of $\Delta t$ as
\begin{align}
    e^{\Delta t{\cal L}^\dagger} &= \one + \Delta t {\cal L}^\dagger + {\cal O}( (\Delta t)^2 )\\
    (\gamma_{t+\Delta t})^{-\frac{1}{2}} &=
    \left( \gamma_t + \Delta t {\cal L}\gamma_t + + {\cal O}( (\Delta t)^2 ) \right)^{-\frac{1}{2}}\\
    &= \gamma_t^{-\frac{1}{2}} - \Delta t \gamma_t^{-\frac{1}{2}} A \gamma_t^{-\frac{1}{2}} + {\cal O}( (\Delta t)^2 ),
\end{align}
where $A = \int_0^\infty e^{-x \sqrt{\gamma_t}} ({\cal L}\gamma_t) e^{-x \sqrt{\gamma_t}} dx $. This leads to
\begin{align}
    \hat{{\cal N}}^{(\texttt{P})}_{\Delta t, \gamma_t}(\rho_t) 
    & = \rho_t + \Delta t  \left[ i \gamma_t^{\frac{1}{2}} H \gamma_t^{-\frac{1}{2}} \rho_t - i  \rho_t \gamma_t^{-\frac{1}{2}} H \gamma_t^{\frac{1}{2}}  + \sum_\mu (\gamma_t^{\frac{1}{2}}L^\dagger_\mu \gamma_t^{-\frac{1}{2}} ) \rho_t (\gamma_t^{-\frac{1}{2}} L_\mu \gamma_t^{\frac{1}{2}} ) \right. \nonumber\\
    &\qquad  \left. -\frac{1}{2} \sum_\mu \left( \gamma_t^{\frac{1}{2}} L^\dagger_\mu L_\mu \gamma_t^{-\frac{1}{2}} \rho_t + \rho_t \gamma_t^{-\frac{1}{2}} L^\dagger_\mu L_\mu \gamma_t^{\frac{1}{2}}\right)  -\left( A \gamma_t^{-\frac{1}{2}}  \rho_t + \rho_t \gamma_t^{-\frac{1}{2}} A \right)  \right] + {\cal O}\left( (\Delta t)^2 \right).
\end{align}

Meanwhile, the Lindbladian dynamics from $\hat{{\cal L}}_{\gamma_t}^{(\texttt{P})}$ becomes
    \begin{align}
        &e^{ \Delta t \hat{{\cal L}}^{(\texttt{P})}_{\gamma_t}}(\rho_t) \\
        &= \left(\one + \Delta t \hat{{\cal L}}^{(\texttt{P})}_{\gamma_t}\right)\rho_t + {\cal O}( (\Delta t)^2 ) \\
        &= \rho_t + \Delta t \bigg[-i (-H + H_{\texttt{C}}) \rho_t + i \rho_t (-H + H_{\texttt{C}}) +  \sum_\mu (\gamma_t^{\frac{1}{2}}L^\dagger_\mu \gamma_t^{-\frac{1}{2}} ) \rho_t (\gamma_t^{-\frac{1}{2}}L_\mu \gamma_t^{\frac{1}{2}} ) - \frac{1}{2} \sum_\mu \left( \gamma_t^{-\frac{1}{2}} L_\mu \gamma_t L^\dagger_\mu \gamma_t^{-\frac{1}{2}} \rho_t + \rho_t \gamma_t^{-\frac{1}{2}} L_\mu \gamma_t L^\dagger_\mu \gamma_t^{-\frac{1}{2}} \right) \bigg] \nonumber\\
        &\qquad + {\cal O}( (\Delta t)^2 ). 
    \end{align}
We then have
\begin{equation}
    \left(\hat{{\cal N}}^{(\texttt{P})}_{\Delta t, \gamma_t} - e^{ \Delta t \hat{{\cal L}}^{(\texttt{P})}_{\gamma_t}}\right)(\rho_t) = \Delta t  \left( B \rho_t + \rho_t B^\dagger \right) + {\cal O}\left( (\Delta t)^2 \right),
\end{equation}
where
\begin{equation}
\begin{aligned}
        B &:=  i \gamma_t^{\frac{1}{2}} H \gamma_t^{-\frac{1}{2}} -\frac{1}{2} \gamma_t^{\frac{1}{2}} \left( \sum_\mu L^\dagger_\mu L_\mu \right) \gamma_t^{-\frac{1}{2}} - A \gamma_t^{-\frac{1}{2}} -i H + i H_{\texttt{C}} + \frac{1}{2} \sum_\mu \gamma_t^{-\frac{1}{2}} L_\mu \gamma_t L^\dagger_\mu \gamma_t^{-\frac{1}{2}}.
\end{aligned}
\end{equation}
To complete the proof, we show that $B=0$. To this end, we first evaluate $A$ as follows:
\begin{align}
    A &= \int_0^\infty e^{-x \sqrt{\gamma_t}} ({\cal L}\gamma_t) e^{-x \sqrt{\gamma_t}} dx  \\
    & = \int_0^\infty \sum_{\lambda_t, \lambda_t'} e^{-x\left(\sqrt{\lambda_t} + \sqrt{\lambda_t'} \right)} \bra{\lambda_t} ({{\cal L}} \gamma_t) \ket*{\lambda_t'}  \ket{\lambda_t}\bra*{\lambda_t'} \\
    & = \sum_{\lambda_t, \lambda_t'} \frac{\bra{\lambda_t}
    ({{\cal L}}\gamma_t) \ket*{\lambda'_t} }{\sqrt{\lambda_t} + \sqrt{\lambda_t'}} \ket{\lambda_t}\bra*{\lambda_t'} \\
    &= -i \sum_{\lambda_t, \lambda_t'} \frac{\bra{\lambda_t}
    [H, \gamma_t] \ket*{\lambda'_t} }{\sqrt{\lambda_t} + \sqrt{\lambda_t'}} \ket{\lambda_t}\bra*{\lambda_t'} + \sum_{\lambda_t, \lambda_t'} \sum_\mu \frac{\bra{\lambda_t}
    L_\mu \gamma_t L^\dagger_\mu \ket*{\lambda'_t} }{\sqrt{\lambda_t} + \sqrt{\lambda_t'}} \ket{\lambda_t}\bra*{\lambda_t'} - \frac{1}{2} \sum_{\lambda_t, \lambda_t'} \sum_\mu  \frac{ (\lambda_t + \lambda_t') \bra{\lambda_t} L^\dagger_\mu L_\mu \ket*{\lambda'_t} }{\sqrt{\lambda_t} + \sqrt{\lambda_t'}} \ket{\lambda_t}\bra*{\lambda_t'},
\end{align}
where we used $\gamma_t \stackrel{s.d.}{=} \sum_{\lambda_t} \op*{\lambda_t}$.

By evaluating all the terms in the eigenbasis of $\gamma_t$, we obtain
\begin{equation}
\begin{aligned}
    \bra {\lambda_t} B \ket*{\lambda_t'} =& i \left( \sqrt{\frac{\lambda_t}{\lambda_t'}} - \frac{\lambda_t - \lambda_t'}{(\sqrt{\lambda_t} + \sqrt{\lambda_t'})\sqrt{\lambda_t'}} - 1\right) \bra{\lambda_t} H \ket*{\lambda'_t} \\
    +&\frac{1}{2} \sum_\mu  \left( - \sqrt{\frac{\lambda_t}{\lambda_t'}} + \frac{\lambda_t + \lambda_t'}{\sqrt{\lambda_t'}(\sqrt{\lambda_t} + \sqrt{\lambda_t'})} +  \frac{\sqrt{\lambda_t} - \sqrt{\lambda_t'}}{\sqrt{\lambda_t} + \sqrt{\lambda_t'}} \right) \bra{\lambda_t} L^\dagger_\mu L_\mu \ket*{\lambda'_t}  \\
    +& \sum_\mu  \left( - \frac{\sqrt{\lambda_t} }{\sqrt{\lambda_t} + \sqrt{\lambda_t'}} +  \frac{1}{2} \frac{\sqrt{\lambda_t} - \sqrt{\lambda_t'}}{\sqrt{\lambda_t} + \sqrt{\lambda_t'}} + \frac{1}{2} \right) \frac{\bra{\lambda_t} L_\mu \gamma_t L^\dagger_\mu \ket*{\lambda'_t}}{\sqrt{\lambda_t} \sqrt{\lambda_t'}}    \\
    =& 0,
\end{aligned}
\end{equation}
for all elements of $\bra{\lambda_t} H \ket*{\lambda'_t}$, $\bra{\lambda_t} \sum_\mu L^\dagger_\mu L_\mu \ket*{\lambda'_t}$, and $\bra{\lambda_t} \sum_\mu L_\mu \gamma_t L^\dagger_\mu \ket*{\lambda'_t}$. 

Consequently, as $\left(\hat{{\cal N}}^{(\texttt{P})}_{\Delta t, \gamma_t} - e^{ \Delta t \hat{{\cal L}}^{(\texttt{P})}_{\gamma_t}}\right)(\rho_t) =  {\cal O}\left( (\Delta t)^2 \right)$ for any $\rho_t$, the diamond norm also scales as
    \begin{align}
        \frac{1}{2} \left\| \hat{{\cal N}}^{(\texttt{P})}_{\Delta t, \gamma_t} - e^{ \Delta t \hat{{\cal L}}^{(\texttt{P})}_{\gamma_t}} \right\|_\diamond &\equiv \frac{1}{2} \max_\sigma \left\|\left(id \otimes \hat{{\cal N}}^{(\texttt{P})}_{\Delta t, \gamma_t} - id \otimes e^{ \Delta t \hat{{\cal L}}^{(\texttt{P})}_{\gamma_t}} \right)(\sigma) \right\|_1 \\
        & = {\cal O}((\Delta t)^2),
    \end{align}
which completes the proof. 
\end{proof}

\subsection{Proof of \cref{rst:Lindblad}} \label{app:Lindblad_proof}

To prove \cref{rst:Lindblad}, we show that (i) 
\begin{align}
        \frac{1}{2}\norm{\hat{{\cal N}}^{(\texttt{P})}_{\Delta t,\gamma_t} - e^{\Delta t \hat{\mathcal{L}}^{(\texttt{T})}_{\xi'}}}_{\diamond}=\order{(\Delta t)^2}, \label{eq:app_pp}
\end{align}
if 

\begin{align}
        &
    H_{\texttt{Ls}}(\xi) - H_{\texttt{C}} = H_{\texttt{Ls}}(\xi') + \order{\Delta t}, 
    \label{eq:app_lindblad_cond1}\\
    &\sum_{jk}p_k\mathcal{D}[\gamma_t^{1/2}\bra*{e_k}_{\texttt{E}}H_{\texttt{tot}}\ket*{e'_j}_{\texttt{E}}\gamma_t^{-1/2}] =\sum_{jk}p'_j\mathcal{D}[\bra*{e_k}_{\texttt{E}}H_{\texttt{tot}}\ket*{e'_j}_{\texttt{E}}] + \order{\Delta t}, 
    \label{eq:app_lindblad_cond2}
\end{align}
for some real number $\alpha$, and (ii) $H_{\texttt{C}}$ vanishes if 
\begin{eqnarray}
    \begin{aligned}
        \gamma_t^{1/2}\bra*{e_k}_{\texttt{E}}{H_{\texttt{tot}}}\ket*{e'_j}_{\texttt{E}}\gamma_t^{-1/2}=c_{jk}\bra*{e_k}_{\texttt{E}}{H_{\texttt{tot}}}\ket*{e'_j}_{\texttt{E}}, \label{eq:Baker}
    \end{aligned}
    \end{eqnarray}
    for some number $c_{jk}$. The correction Hamiltonian $H_{\texttt{C}}=H_{\texttt{C},\gamma_t}$ is defined as in Eq.~\eqref{eq:hc}.\\

\noindent (i) Once we notice that Eq.~\eqref{eq:app_lindblad_cond1} and Eq.~\eqref{eq:app_lindblad_cond2} imply  $\hat{\mathcal{L}}^{(\texttt{P})}_{\gamma}=\hat{\mathcal{L}}^{(\texttt{T})}_{\xi'}+ \order{\Delta t}$, it is immediate from \cref{lem:Petz_Lindblad} that Eq.~\eqref{eq:app_pp} holds.\\

\noindent (ii) Suppose 
\begin{align}
    \gamma_t^{1/2}L_{jk}^\dagger\gamma_t^{-1/2}=c_{jk}L_{jk}^\dagger,
\end{align}
where we denote $L_{jk}^\dagger= \bra*{e_k}_{\texttt{E}}{H_{\texttt{tot}}}\ket*{e'_j}_{\texttt{E}}$. This implies
\begin{align}
    \comm{\ln\gamma_t}{L_{jk}^\dagger}=(\ln{c_{jk}})L_{jk}^\dagger.
\end{align}

Consequently, we also have $\comm{\ln\gamma_t}{L_{jk}}=(-\ln{c_{jk}})L_{jk}$. They collectively imply that 
\begin{align}
    \begin{aligned}
        \comm{\ln\gamma_t}{L_{jk}\+L_{jk}} &= (\ln\gamma_t L_{jk}\+)L_{jk} - L_{jk}\+(L_{jk} \ln\gamma_t) \\
        &= \left(\comm{\ln\gamma_t}{L_{jk}\+}+L_{jk}\+\ln\gamma_t\right)L_{jk} - L_{jk}\+\left(\ln\gamma_t L_{jk} - \comm{\ln\gamma_t}{L_{jk}}\right)\\
        &= \left((\ln c_{jk}) L_{jk}\++L_{jk}\+\ln\gamma_t\right)L_{jk} - L_{jk}\+\left(\ln\gamma_t L_{jk} + (\ln c_{jk})L_{jk}\right)\\
        &=0.
    \end{aligned}
    \end{align}

We also evaluate $\comm{\ln\gamma_t}{\gamma_t^{-\frac{1}{2}}L_{jk}\gamma_t^{\frac{1}{2}}}$ as follows,
\begin{align}
    \comm{\ln\gamma_t}{\gamma_t^{-\frac{1}{2}}L_{jk}\gamma_t^{\frac{1}{2}}}&=\gamma_t^{-\frac{1}{2}}\comm{\ln\gamma_t}{L_{jk}}\gamma_t^{\frac{1}{2}}\\
    &=(-\ln{c_{jk}})\gamma_t^{-\frac{1}{2}}L_{jk}\gamma_t^{\frac{1}{2}}.
\end{align}
Then, similar to the above evaluation, we get 
\begin{align}
    \comm{\ln\gamma_t}{\gamma_t^{-1/2}L_{jk}\gamma_t L_{jk}\+\gamma_t^{-1/2}}=0.    
\end{align}
Combining altogether, we have 
\begin{align}
    \comm{\ln\gamma_t}{M(\gamma_t)}=0.
\end{align}
Let us write $M(\gamma_t)$ with respect to the eigenbasis of $\gamma_t\stackrel{s.d.}{=} \sum \lambda_t \op*{\lambda_t}$ such that $M(\gamma_t)=\sum_{\lambda_t,\lambda_t'} a_{\lambda_t,\lambda_t'}\ketbra{\lambda_t}{\lambda_t'}$. Then, we note that $\comm{\ln\gamma_t}{M(\gamma_t)}=0$ implies $(\sqrt{\lambda_t}-\sqrt{\lambda_t'}) a_{\lambda_t,\lambda_t'} =0$:
    \begin{align}
        &\comm{\ln\gamma_t}{M(\gamma_t)}=0 \\
        &\Leftrightarrow \sum_{\lambda_t,\lambda_t'} (\ln\lambda_t-\ln\lambda_t') a_{\lambda_t,\lambda_t'}\ketbra{\lambda_t}{\lambda_t'} =0 \\
        & \Leftrightarrow (\ln\lambda_t-\ln\lambda_t') a_{\lambda_t,\lambda_t'} =0 \quad \forall\, \lambda_t,\lambda_t'\\
        & \Leftrightarrow (\lambda_t-\lambda_t') a_{\lambda_t,\lambda_t'} =0 \quad \forall\, \lambda_t,\lambda_t'\\
        & \Leftrightarrow (\sqrt{\lambda_t}-\sqrt{\lambda_t'}) a_{\lambda_t,\lambda_t'} =0 \quad \forall\, \lambda_t,\lambda_t'.
\end{align}
Finally, we can conclude that $H_{\texttt{C},\gamma_t}=0$:
    \begin{align}
        H_{\texttt{C},\gamma_t}
        &\equiv\frac{1}{2i}\sum_{\lambda_t,\lambda_t'}\left(\frac{\sqrt{\lambda_t}-\sqrt{\lambda_t'}}{\sqrt{\lambda_t}+\sqrt{\lambda_t'}}\right)\bra{\lambda_t}M(\gamma_t) \ket{\lambda_t'} \op{\lambda_t}{\lambda_t'}\\
        &= \frac{1}{2i}\sum_{\lambda_t,\lambda_t'}\left(\frac{\sqrt{\lambda_t}-\sqrt{\lambda_t'}}{\sqrt{\lambda_t}+\sqrt{\lambda_t'}}\right) a_{\lambda_t,\lambda_t'} \op{\lambda_t}{\lambda_t'}\\
        &= \frac{1}{2i}\sum_{\lambda_t,\lambda_t'}\frac{1}{\sqrt{\lambda_t}+\sqrt{\lambda_t'}} \left(\underbrace{(\sqrt{\lambda_t}-\sqrt{\lambda_t'})a_{\lambda_t,\lambda_t'}}_{=0}\right) \op{\lambda_t}{\lambda_t'}\\
        &=0.
    \end{align}

\bibliography{apssamp.bib}

\end{document}